\documentclass[acmsmall,10pt,screen]{acmart}

\usepackage{xspace}
\usepackage{listings}
\usepackage{xypic}
\usepackage{mathpartir}
\usepackage{stmaryrd}
\usepackage{longtable}
\usepackage{subcaption}
\usepackage{diagbox}
\usepackage{enumitem}
\usepackage{mathtools}

\newcommand{\name}{\textbf{FOOD}\xspace}
\newcommand{\systemname}{\textsc{\textbf{Cook}}\xspace}

\definecolor{light-gray}{gray}{0.85}

\newcommand{\an}[3]{}

\newcommand{\cristina}[1]{\an{cristina}{green}{#1}}


 %TODO mandatory; please add comma-separated list of keywords

\begin{document}

\renewcommand{\sectionautorefname}{Section}
\renewcommand{\subsectionautorefname}{Section}

\title{Decomposition Without Regret}

% \keywords{Expression Problem, Compositionality, Traits}

\author{Weixin Zhang}
\author{Cristina David}
\author{Meng Wang}
\affiliation{\institution{University of Bristol} \city{Bristol} \country{United Kingdom}}
\email{{weixin.zhang, cristina.david, meng.wang}@bristol.ac.uk}

%\EventEditors{John Q. Open and Joan R. Access}
%\EventNoEds{2}
%\EventLongTitle{42nd Conference on Very Important Topics (CVIT 2016)}
%\EventShortTitle{CVIT 2016}
%\EventAcronym{CVIT}
%\EventYear{2016}
%\EventDate{December 24--27, 2016}
%\EventLocation{Little Whinging, United Kingdom}
%\EventLogo{}
%\SeriesVolume{42}
%\ArticleNo{23}
%%%%%%%%

\begin{abstract}
Programming languages are embracing both functional and object-oriented paradigms. 
A key difference between the two paradigms is the way of achieving data abstraction. 
That is, how to organize data with associated operations.  
There are important tradeoffs between functional and object-oriented decomposition in terms of extensibility and expressiveness. 
Unfortunately, programmers are usually forced to select a particular decomposition style in the early stage of programming.  
Once the wrong design decision has been made, the price for switching to the other decomposition style could be rather high since pervasive manual refactoring is often needed. 
To address this issue, this paper presents a bidirectional transformation system between functional and object-oriented decomposition. 
We formalize the core of the system in the \name calculus, which captures the essence of functional and object-oriented decomposition.
We prove that the transformation preserves the type and semantics of the original program. 
We further implement \name in Scala as a translation tool called \systemname and conduct several case studies to demonstrate the applicability and effectiveness of \systemname.
\end{abstract}

\maketitle
\section{Introduction}\label{sec:intro}
Programming languages are embracing multiple paradigms, in particular functional and object-oriented paradigms.
Modern languages are designed to support multi-paradigms. Well-known examples are OCaml, Swift,
Rust, TypeScript, Scala, F\#, Kotlin, to name a few. % R(Evaluating the Design of the R Language)
Meanwhile, mainstream object-oriented languages such as C++ and Java are gradually extended
to support functional paradigms.  When multiple paradigms are available within one programming language, a natural question to ask is: 
\emph{which paradigm should a programmer choose when designing programs}?
% Haskell’s overlooked object system

% destruction vs construction
A fundamental difference between functional and object-oriented paradigms is
the way of achieving \emph{data abstraction}~\citep{Reynolds78,cook2009}. That is, how to organize data with associated operations.  
Typically, object-oriented decomposition is \emph{operation first}: we first declare an
interface that describes the operations supported by the data
and then implement that interface with some classes.  Conversely, functional decomposition is \emph{data first}: we first represent the data 
using an algebraic datatype and then define operations by pattern matching on that algebraic datatype.

\autoref{decomposition} manifests the difference between object-oriented and functional decomposition by implementing an evaluation operation on literals and subtractions 
with both styles in Scala. The object-oriented decomposition version shown on the left hand side models expressions as a class hierarchy,
where the interface \lstinline{Exp} declares an abstract \lstinline{eval} method and two classes \lstinline{Lit} and \lstinline{Sub} concretely implement \lstinline{eval}. 
An \lstinline{Exp} object is created via calling \lstinline{new} on the classes and then evaluated by invoking the \lstinline{eval} method.
In contrast, the functional decomposition version shown on the right-hand side defines \lstinline{Exp} as an algebraic datatype with 
\lstinline{Lit} and \lstinline{Sub} being the constructors.  The \lstinline{eval} operation is separately defined as a pattern matching function on \lstinline{Exp}.
An \lstinline{Exp} instance is created by calling the constructors and then evaluated by applying the \lstinline{eval} function.

\begin{figure*}[t]
\begin{minipage}{.48\textwidth}
\begin{lstlisting}[linerange=3-14]
// class hierarchy
trait Exp { 
  def eval: Int 
}
class Lit(n: Int) extends Exp {
  def eval: Int = n
}
class Sub(l: Exp, r: Exp) extends Exp {
  def eval: Int = l.eval - r.eval
}
// client code
new Sub(new Lit(1),new Lit(2)).eval
\end{lstlisting}%APPLY:OOP
\end{minipage}
\begin{minipage}{.51\textwidth}
\begin{lstlisting}[linerange=3-14]
// algebraic datatype
sealed trait Exp
case class Lit(n: Int) extends Exp
case class Sub(l: Exp, r: Exp) extends Exp

// pattern matching function
def eval(exp: Exp): Int = exp match {
  case Lit(n)   => n
  case Sub(l,r) => eval(l) - eval(r) 
}
// client code
eval(Sub(Lit(1),Lit(2)))
\end{lstlisting}%APPLY:FP
\end{minipage}
\caption{Object-oriented decomposition (left) vs. functional decomposition (right) in Scala.}
\label{decomposition}
\end{figure*}

% two \emph{constructors} for zero (\lstinline{Zero}) and successors (\lstinline{Succ}).
% The \lstinline{eval} function is defined outside of the case class hierarchcy
% by pattern matching over values of \lstinline{Exp} and calling itself recursively
% on the child nodes.

% \emph{data abstraction}~\cite{Reynolds78,Cardelli85,Cook09abstraction}

% the tradeoff between these two styles
There are important tradeoffs between functional and object-oriented decompositions
in terms of extensibility and expressiveness. %have their own strengths and weaknesses.
As acknowledged by the notorious Expression Problem~\citep{Reynolds78,cook1990object,expPb}, these two decomposition
styles are complementary in terms of \emph{extensibility}. Object-oriented
decomposition makes it easy to extend data variants through defining new classes. 
For example, negations can be added to the \lstinline{Exp} hierarchy modularly:
\begin{lstlisting}[linerange=17-17]
class Neg(e: Exp) extends Sub(new Lit(0), e)
\end{lstlisting}%APPLY:OOP_NEG
On the other hand, functional decomposition makes it easy to add new operations such as simplification 
on expressions: 
\begin{lstlisting}[linerange=18-21]
def simplify(exp: Exp): Exp = exp match {
  case Sub(l,Lit(0)) => l
  case _             => exp
}
\end{lstlisting}%APPLY:SIMPLIFY
% Extensions on \lstinline{Neg} and \lstinline{simplify} also illustrate the difference between
Besides extensibility, object-oriented and functional decomposition have different expressive power. 
Object-oriented decomposition facilitates code reuse through 
inheritance (e.g., \lstinline{Neg}) and enables \emph{interoperability}~\citep{jonathan2013} 
between different implementations of the same interface whereas functional decomposition allows inspection on the internal 
representation of data through (nested) pattern matching, simplifying abstract syntax tree transformations (e.g., \lstinline{simplify}).

Unfortunately, programmers are forced to decide a particular decomposition
style in the early stage of programming. A proper choice, however,
requires predictions on the extensibility dimension and kinds of operations to model,
which may not be feasible in practice. Once the wrong design decision has been made, 
the price for switching to the other decomposition style could be rather high 
since pervasive manual refactoring is often needed. For example, to convert the 
left-hand side of \autoref{decomposition} to its right hand side, one has to: 
1) change \lstinline{class}es to \lstinline{case class}es; 
2) extract \lstinline{eval} out of the class hierarchy, add an argument of type \lstinline{Exp}, and  
 merge the implementations using \lstinline{case} clauses; 
3) switch method selections to function applications, e.g. \lstinline{l.eval} to \lstinline{eval(l)};
4) remove \lstinline{new}s on \lstinline{Lit} and \lstinline{Sub}. 
Such manual transformation from object-oriented decomposition to functional decomposition is rather tedious and error-prone and so as for the other way round.

A better way, however, allows programmers to choose a decomposition style for 
prototyping without regret. When the design choice becomes inappropriate, a tool automatically
transforms their code into the other style without affecting the semantics. 
The need for switching decomposition styles actually happens in the real world. 
For example, to grow their simple SQL processor into a realistic SQL engine, \citet{rompf2019sql} switch the decomposition style from functional to object-oriented
for supporting a large number of operators. 
Even at later stages, such a automatic translation tool could be used to make extensions of data variants or operations easier
by momentarily switching the decomposition, adding the extension, and then transforming the program back to the original decomposition.
Furthermore, studying the transformation between the two styles can provide a theoretical foundation for compiling multi-paradigm languages into single-paradigm ones.
From an educational perspective, the tool can help novice programmers to understand both decomposition styles better.

%bijective
To address this issue, we propose a bidirectional transformation between
functional and object-oriented decomposition based on the observation that restricted forms
of functional and object-oriented decomposition are \textit{symmetric}. 
Our bidirectional transformation is greatly inspired by \citet{cook2009}'s pure object-oriented programming
and the line of work on the duality of data and codata~\cite{rendel2015,klaus2018,downen2019codata,popl2019}.
The main novelty of our work is an automatic, type-directed transformation formalized in the core calculus \name, which captures
the essence of Functional and Object-Oriented Decomposition.
Unlike existing work, \name does not require new language design and can be directly applied to existing multi-paradigm languages.
We implement \name in Scala as a translation tool called \systemname
and conduct several case studies to demonstrate the applicability of \systemname.
% \name is implemented in Scala.
% Moreover, the type-directed transformation approach adopted by \name can properly handle many object-oriented programming features that previous approaches cannot.
% Although there is an inherent connection between codata and objects, existing syntax-directed transformation approaches are not applicable to existing multi-paradigm languages

% Inheritance, subtyping, mutable state are orthogonal to codata
% it combines data with operations under a formalized interface.
% functions are first-class
% decomposition - breaking complex systems into components
%abstract class/interfaces, interactions, composition/inheritance, dynamic binding/polymorphism
In summary, the contributions of our work are:
\begin{itemize}
\item We introduce functional and object-oriented decomposition 
and show their correspondence (\autoref{sec:overview}).
\item We formalize the type-directed bidirectional transformation between functional and object-oriented decomposition in the \name calculus 
and show how to concretely transform programs using \name (\autoref{sec:formalization}).
\item We give the semantics of \name and prove that the transformation preserves the type and semantics of the original program(\autoref{sec:soundness}).
\item We implement the approach in Scala as a translation tool called \systemname, and conduct several case studies to demonstrate its applicability (\autoref{sec:case}).
\item We discuss features beyond \name and sketch how \name can be ported to other multi-paradigm languages (\autoref{sec:discussion}).
% available online at~
% \footnote{Due to the anonymous submission, we submit these as supplementary materials.}:
\end{itemize}

Examples, case studies, complete rules, proofs, and the implementation of \systemname can be found in the \textbf{supplementary material}.

\section{Background}\label{sec:overview}
In this section, we introduce functional and object-oriented decomposition and show their correspondence.
To facilitate the discussion, we reuse the integer sets presented by \citet{cook2009} as a running example.
% constructorization/destructorization
% terminology: transpose/transposition

\begin{figure}[t]
% \noindent\begin{minipage}{.4\textwidth}
% \begin{lstlisting}
% trait Set {
%   def isEmpty: Boolean
%   def contains(i: Int): Boolean
%   def insert(i: Int): Set = 
%     if (this.contains(i)) this
%     else new Insert(this,i)
%   def union(s: Set): Set = new Union(this,s)
% }

% object Empty extends Set {
%   def isEmpty = true
%   def contains(i: Int) = false
%   override def union(s: Set) = s
% }
% \end{lstlisting}
% \end{minipage}\hfill
% \begin{minipage}{.6\textwidth}
% \begin{lstlisting}
% class Insert(s: Set, n: Int) extends Set {
%   def isEmpty = false
%   def contains(i: Int) = i == n || s.contains(i)
% }

% class Union(s1: Set, s2: Set) extends Set {
%   def isEmpty = s1.isEmpty && s2.isEmpty
%   def contains(i: Int) = s1.contains(i) || s2.contains(i)
% }

% Empty.insert(1).union(Empty.insert(3)).contains(3)
% \end{lstlisting}
% \end{minipage}
\begin{lstlisting}[linerange=3-28]
trait Set {
  def isEmpty: Boolean
  def contains(i: Int): Boolean
  def insert(i: Int): Set = 
    if (this.contains(i)) this
    else new Insert(this,i)
  def union(s: Set): Set = new Union(this,s)
}

object Empty extends Set {
  def isEmpty = true
  def contains(i: Int) = false
  override def union(s: Set) = s
}

class Insert(s: Set, n: Int) extends Set {
  def isEmpty = false
  def contains(i: Int) = i == n || s.contains(i)
}

class Union(s1: Set, s2: Set) extends Set {
  def isEmpty = s1.isEmpty && s2.isEmpty
  def contains(i: Int) = s1.contains(i) || s2.contains(i)
}

Empty.insert(1).union(Empty.insert(3)).contains(3)
\end{lstlisting}%APPLY:SET_OOP
\caption{Object-oriented decomposition for integer sets}
\label{setoop}
\end{figure}

\subsection{Object-oriented decomposition}

% \paragraph*{Pure object-oriented programming}
The object-oriented decomposition style used throughout this paper follows
\citet{cook2009}'s definition of pure object-oriented programming (OOP), which is arguably the essence of OOP.
Here are the key principles of pure OOP:
\begin{enumerate}
  \item Objects are first-class values;
  \item Interfaces are used as types not classes;
  \item Objects are accessed through their public interfaces
\end{enumerate}
In pure OOP, a class is an object \emph{generator}, which is a procedure that returns a value satisfying an interface.
The internal representation of an object is hidden to other objects.
Although most OOP languages do not strictly follow the pure OOP style, programming in pure OOP style brings several advantages,
resulting in more flexible software systems.
More importantly, as we shall see later, it allows the code to be transformed to functional decomposition style.
% Cook “an object is a value exporting a procedural interface to data or behavior"
% Such autognosis brings flexibility but make it hard to model optimizations.
% object interfaces do not use type abstraction

\paragraph*{Interfaces}
A pure OOP implementation of integer sets is given in~\autoref{setoop}.
The interface \lstinline{Set} describes four operations supported by an integer set: 
\lstinline{isEmpty} checks whether a set is empty; 
\lstinline{contains} tells whether an integer is in the set;
\lstinline{insert} puts an integer to a set if it is currently not in that set; 
\lstinline{union} merges two sets.
These methods are often referred as \textit{destructors}, which tear down an object to some other type.
Special destructors whose return types are \lstinline{Set} (e.g. \lstinline{insert} and \lstinline{union}) 
are also known as \emph{producers} or \emph{mutators}.
% \lstinline{insert} and \lstinline{union} have a default implementation that uses the corresponding class 

% \paragraph*{Interface polymorphism} %inclusion polymorphism %subtype polymorphism
\paragraph*{Classes}
The \lstinline{Set} interface is implemented by three classes \lstinline{Empty}, \lstinline{Insert}, and \lstinline{Union}.
The singleton class \lstinline{Empty} (\lstinline{object} in Scala) represents an empty set;
The \lstinline{Insert} class models the insertion of an integer to a set;
The \lstinline{Union} class models the union of two sets. 
As shown by \lstinline{insert} and \lstinline{union}, the sole use of classes is in \lstinline{new} expressions for creating objects, following the principles of pure OOP.

\paragraph*{Inheritance and method overriding}
Although inheritance is not an OOP-specific feature~\cite{CookPalsberg}, it is commonly available in OOP languages for code reuse.
To avoid duplicating the same implementation in different classes, both \lstinline{insert} and \lstinline{union} come with a default implementation in the \lstinline{Set} interface.
These default implementations are inherited by \lstinline{Set}'s classes so that explicit definitions for \lstinline{insert} and \lstinline{union} can be omitted.
An exception is \lstinline{Empty}, which \emph{overrides} the \lstinline{union} destructor to return the other set immediately instead 
of creating a \lstinline{Union} object. Through \emph{dynamic dispatching}, the overridden definitions will be selected instead of the inherited ones.
% fine-grained code reuse.

% Now, following the principle of ``composition over inheritance'' brings one extra benefit---keeping the code free from decomposition style.

% may consider harmful.  These forms of inheritance can be encoded into delegations.

\paragraph*{This}
A special variable \lstinline{this} (\lstinline{self} or \lstinline{me} in other OOP languages) is in the scope of an interface or a class
for referring to the object itself.
As illustrated by the default implementation of \lstinline{insert}, \lstinline{this} is used for constructing a (\lstinline{new Insert(this,i)}))
or invoking other methods (can often be omitted \lstinline{contains(i)}).

% \paragraph*{Complex hierarchies}
% Traditional algebraic datatypes are \emph{flat}, where variants have no relationship among each other. 
% On the other hand, can be \emph{hierarchical} where a class may not directly extend the root but some intermediate class or datatype~\cite{castor}. 
% This way of organizing classes promotes finer-grained reuse. Naively transforming hierarchical 
% datatypes is problematic since case classes are \emph{final}, i.e. extending a case class

There are other features that are often (wrongly) recognized as OOP features, such as \emph{mutable state} and \emph{subtyping}.
As \citet{cook2009} argues, these features are not essentials of OOP. 
We will discuss these features in \autoref{sec:morefeatures}. 
% And more importantly, they do not have a clear correspondence in functional decomposition. 
% In other words, using these features will fix the decomposition style to be object-oriented.

% \emph{private methods}
% There are arguments. Private methods are code smells: they violate single responsibility principle;
% breaks abstraction; making the program untestable.
% Now, there is one more reason to 

\begin{figure}[t]
\begin{lstlisting}[linerange=3-29]
sealed trait Set 
case object Empty extends Set
case class Insert(s: Set, n: Int) extends Set
case class Union(s1: Set, s2: Set) extends Set

def isEmpty(self: Set): Boolean = self match {
  case Empty        => true
  case Insert(s,n)  => false
  case Union(s1,s2) => isEmpty(s1) && isEmpty(s2)
}

def contains(self: Set)(i: Int): Boolean = self match {
  case Empty        => false
  case Insert(s,n)  => n == i || contains(s)(i)
  case Union(s1,s2) => contains(s1)(i) || contains(s2)(i)
}

def insert(self: Set)(i: Int): Set =
  if (contains(self)(i)) self
  else Insert(self,i)

def union(self: Set)(that: Set): Set = self match {
  case Empty => that
  case _     => Union(self,that)
}

contains(union(insert(Empty)(1))(insert(Empty)(3)))(3)
\end{lstlisting}%APPLY:SET_FP
\caption{Functional decomposition for integer sets}
\label{setfp}
\end{figure}

\subsection{Functional decomposition}
% TODO: definition of functional decomposition
The essence of functional decomposition is arguably algebraic datatypes plus pattern matching functions.
A functional decomposition implementation of integer sets is given in \autoref{setfp}.

\paragraph*{Algebraic datatypes}
Now, \lstinline{Set} is modeled as an \emph{algebraic datatype} (\lstinline{sealed trait} in Scala) with \lstinline{Empty}, \lstinline{Insert} and \lstinline{Union} 
captured as \emph{constructors} (\lstinline{case class} in Scala) of \lstinline{Set}. Unlike classes, a \lstinline{Set} object is created by calling the constructors without \lstinline{new}, e.g. \lstinline{Insert(Empty,3)}.

\paragraph*{Pattern matching functions}
Four functions, \lstinline{isEmpty}, \lstinline{contains}, \lstinline{insert}, and \lstinline{union} are defined on the \lstinline{Set} datatype.
These functions are called \emph{consumers} on \lstinline{Set} since their first argument \lstinline{self} is a datatype.
A consumer is typically defined through \emph{pattern matching} (\lstinline{match} clause in Scala) on the datatype and giving a \lstinline{case} clause 
for each constructor of that datatype. For example, \lstinline{isEmpty} and \lstinline{contains} are implemented in this manner. 

\paragraph*{Wildcard pattern}
Sometimes it is tedious to repeat case clauses for constructors that have a common behavior.  
For such cases, the wildcard pattern is often used for giving a default behavior to all boring cases at once. 
For instance, the \lstinline{union} consumer uses a wildcard pattern, \lstinline{case _ => Union(self,that)},
to handle non-\lstinline{Empty} cases altogether. 

% \paragraph*{Guarded patterns}

% \subsection{Restricted pattern matching}
% “disentanglement” unnesting
% Ager et al. 2003. Setzer et al. (2014)
% \paragraph*{Pattern with guards}

\begin{comment}
\begin{quote}
It is not clear what nested pattern matching or pattern matching on multiple arguments should be refunctionalized to.
Fortunately, the restriction to top-level pattern matching on the first argument does not restrict the expressiveness of the language
because we can desugar nested pattern matching or pattern matching on multiple arguments by introducing helper functions.
\end{quote}
\end{comment}

\begin{figure}
  \centering
  \begin{tabular}{rcl}
\textbf{Object-oriented} & & \textbf{Functional} \\
Interface & $\Longleftrightarrow$ & Datatype \\
Destructor & $\Longleftrightarrow$ & Consumer \\
Generator & $\Longleftrightarrow$ &  Constructor \\
Interface-based inheritance & $\Longleftrightarrow$ & Wildcard pattern\\

% Default method & $\Longleftrightarrow$ & Wildcard pattern\\
  \end{tabular}
  \caption{The correspondence between object-oriented and functional decomposition}
  \label{correspondence}
\end{figure}

\subsection{The correspondence between functional and object-oriented decomposition}
Having a closer look at \autoref{setoop} and \autoref{setfp}, we can see that they are essentially \emph{symmetric}.
Their correspondence is revealed by \autoref{correspondence}, where features in one style has a one-to-one mapping in the other style. 
In other words, it is possible to transform from one decomposition style to the other back and forth without losing information.
\autoref{sec:formalization} will formalize the bidirectional transformation and concretely show how to transform the two versions of integer sets 
discussed here.

\newcommand{\kw}[1]{\texttt{#1}}
\newcommand{\hast}{\!:\!}
\newcommand{\app}[2]{#1 \; #2}
\newcommand{\tapp}[2]{#1 \; @ #2}
\newcommand{\mergeop}{\; ,, \;}
\newcommand{\lam}[3]{\lambda (#1 \hast #2).\binderspacing #3}
\newcommand{\blam}[2]{\Lambda #1.\binderspacing #2}
\newcommand{\proj}[2]{{\code{proj}}_{#1} #2}
\newcommand{\reccon}[2]{\{ #1=#2 \}}
\newcommand{\recty}[2]{\{ #1 \hast #2 \}}
\newcommand{\letexpr}[3]{\kwlet \; #1 = #2 \; \kwin \; #3}
\newcommand{\context}[2]{#1 \; {\vdash} \; #2}

\newcommand{\eval}[2]{#1 \longrightarrow #2}

\newcommand{\sig}{\textsc{Sig}}
\newcommand{\defs}{\textsc{Def}}
\newcommand{\dtrBody}{\textsc{dtrBody}}
\newcommand{\csmBody}{\textsc{csmBody}}
\newcommand{\dtrType}{\textsc{dtrSig}}
\newcommand{\body}{\textsc{body}}
\newcommand{\fun}[2]{\dec{#1} = #2}
\newcommand{\sealed}[1]{\kwdata \; #1}
\newcommand{\inter}[2]{\kwabstract \; \kwclass \; #1 \{ \overline{#2} \}}
\newcommand{\interface}[2]{\kwinterface \; #1 \; \{ \overline{#2}\}}
\newcommand{\new}[2]{\kwnew \; #1(#2)}
\newcommand{\newSeq}[2]{\new #1{\overline{#2}}}
\newcommand{\ctr}[2]{\kwcase \; #1(\overline{x : T}) \; \kwextends \; #2}
\newcommand{\classf}[5]{\kwclass \; #1(\overline{#2 : #3}) \; \kwimplements \; #4 \; \{ #5 \}}
\newcommand{\class}[3]{\classf #1 x T #2 {\overline{#3}}}

\newcommand{\kwinterface}{\kw{interface}}
\newcommand{\kwextends}{\kw{extends}}
\newcommand{\kwimplements}{\kw{implements}}
\newcommand{\kwtrait}{\kw{trait}}
\newcommand{\kwdata}{\kw{data}}
\newcommand{\kwabstract}{\kw{abstract}}
\newcommand{\kwclass}{\kw{class}}
\newcommand{\kwsealed}{\kw{sealed}}
\newcommand{\kwcase}{\kw{case}}
\newcommand{\kwdef}{\kw{def}}
\newcommand{\kwsuper}{\kw{super}}
\newcommand{\kwthis}{\kw{this}}
\newcommand{\kwmatch}{\kw{match}}
\newcommand{\kwoverride}{\kw{override}}

\newcommand{\kwobj}{\kw{obj}}
\newcommand{\obj}[1]{C({\overline{#1}})}
\newcommand{\kwself}{\kw{self}}
\newcommand{\kwnew}{\kw{new}}
\newcommand{\kwarr}{\kw{=>}}
\newcommand{\kwint}{\kw{Int}}

\newcommand{\inbracket}[1]{\llbracket#1\rrbracket}
\newcommand{\dt}{\textsc{Dt}}
\newcommand{\ct}{\textsc{It}}
\newcommand{\dtr}[1]{\textsc{Dtr}(#1)}
\newcommand{\consumer}[1]{\textsc{Csm}(#1)}
\newcommand{\case}[2]{\kwcase\ {#1} \Rightarrow {#2}}
\newcommand{\csmcase}[1]{\kwdef \ f(\kwself:D)(\overline{x : T}): T = \overline{\case P{#1}}}
\newcommand{\csm}[1]{\kwdef \ f(\kwself:D)(\overline{x : T}): T = #1}
\newcommand{\dec}[1]{\kwdef \ f(\overline{#1 : T}): T}
\newcommand{\generator}[1]{\textsc{Gen}(#1)}
\newcommand{\toCase}[2]{\textsc{ToCase}(#1,#2)}
\newcommand{\ctrs}[1]{\textsc{Ctr}(#1)}
\newcommand{\dtrToCsm}[2]{\Delta;\Gamma \vdash_D #1 \leadsto #2}
\newcommand{\dtrToCsmCtx}[4]{\Delta;\Gamma, #1 : #2 \vdash_D #3 \leadsto #4}
\newcommand{\dtrToCase}[2]{\Delta;\Gamma \vdash_f {#1} \leadsto {#2}}
\newcommand{\csmToDtr}[2]{\Delta;\Gamma \vdash_C #1 \leadsto #2}
\newcommand{\csmSig}{\textsf{csm2dtrsig}}
\newcommand{\map}[3]{\langle #1 \; | \; #2 \leftarrow #3 \rangle}

\newcommand{\transCtx}[4]{\Delta;\Gamma, #1 \vdash #2 \Rightarrow #3 \leadsto #4}
\newcommand{\trans}[3]{\Delta; \Gamma \vdash #1 \Rightarrow #2 \leadsto #3}
\newcommand{\transGeneric}[4]{#1 \vdash #2 \Rightarrow #3 \leadsto #4}
\newcommand{\aliasCtx}{\Delta}
\newcommand{\toplike}[1]{\rceil #1 \lceil}
\newcommand{\subtype}[2]{#1 <: #2}
\newcommand{\Tau}{\mathrm{T}}
\newcommand{\all}[3]{\forall(#1*#2).#3}
\newcommand{\pr}[1]{\langle #1 \rangle}
\newcommand{\tid}{\kw{id}}
\newcommand{\ttop}{\kw{top}}
\newcommand{\tbot}{\kw{bot}}
\newcommand{\tdist}{\kw{dist}}
\newcommand{\returntype}{\textsc{returntype}}
\newcommand{\distinguish}{\textsc{distinguish}\xspace}
\newcommand{\sortCtx}{\Sigma}
\newcommand{\desugarType}[4]{\sortCtx \vdash_{#1}^{#2} #3 \Rightarrow #4}
\newcommand{\desugarTypeCtx}[5]{#1 \vdash_{#2}^{#3} #4 \Rightarrow #5}
\newcommand{\ctype}{\textsc{fieldType}}
\newcommand{\forward}{\textasciicircum\xspace}
\newcommand{\exclu}{\backslash}

\newcommand{\selDtr}{
\inferrule[Sel2App]
{ \trans {e_1} D {e_1'} \\ f \in \dtr{D} \\ \dtrType(f,D) = \overline{T} \rightarrow T\\ \overline{\trans {e_2} T {e_2'}}}
{ \trans{e_1.f(\overline{e_2})}{T}{f(e_1')(\overline{e_2'})}}
}
\newcommand{\appCsm}{
\inferrule[App2Sel]
{ \trans {e_1} D {e_1'} \\ f \in \consumer{D} \\ \sig(f) = D \rightarrow \overline{T} \rightarrow T \in \Gamma \\ \overline{\trans {e_2} T {e_2'}}}
{ \trans{f(e_1)(\overline{e_2})}{T}{e_1'.f(\overline{e_2'})}}
}
\newcommand{\appCtr}{
\inferrule[Obj2New]
{ \sig(C) = \overline{T} \rightarrow D \\ C \in \ctrs{D} \\ \overline{\trans e T e'}}
{ \trans{C(\overline{e})}{D}{\new C {e'}}}
}
\newcommand{\newGen}{
\inferrule[New2Obj]
{  \sig(C) = \overline{T} \rightarrow D \\ C \in \generator{D} \\ \overline{\trans e T e'}} %\\ {\trans t {\overline{T} \arrow T} {t'}} }
{ \trans{\newSeq C e}{D}{C(\overline{e'})}}
}

\newcommand{\eCongr}{
\inferrule[E-Congr]
{ \eval {e_1} {e_2}}
{ \eval{E[e_1]}{E[e_2]}}
}

% \inferrule[E-Var]
% { \overline{\eval e v}}
% { \eval{\new C e} obj(C,\overline{v})}

\newcommand{\enew}{
\inferrule[E-New]
{ }
{ \eval{\newSeq{C}{v}} {obj(C,\overline{v})}}
}

\newcommand{\ectr}{
\inferrule[E-Ctr]
{ }
{ \eval{C(\overline{v})}{obj(C,\overline{v})}}
}

\newcommand{\edtrInvk}{
\inferrule[E-Dtr]
{ \dtrBody(f,C) = (\overline{y},\overline{x},e)}
{ \eval{obj(C,\overline{v_1}).f(\overline{v_2})}{[\kwthis \mapsto obj(C,\overline{v_1}),\overline{y} \mapsto \overline{v_1}, \overline{x} \mapsto \overline{v_2} ]e}} 
}

\newcommand{\ecsmInvk}{
\inferrule[E-Csm]
{\csmBody(f,C) = (\overline{y},\overline{x},e)}
{ \eval {f(obj(C,\overline{v_1}))(\overline{v_2})} {[\kwself \mapsto obj(C,\overline{v_1}),\overline{y} \mapsto \overline{v_1}, \overline{x} \mapsto \overline{v_2}]e}}
}
\newcommand{\edtrRecv}{
\inferrule[E-Dtr-Recv]
{ \eval{e_1}{e_1'} }
{ \eval{e_1.f(\overline{e})}{e_1'.f(\overline{e})} }
}

\newcommand{\edtrArg}{
\inferrule[E-Dtr-Arg]
{ \eval{e_i}{e_i'} }
{ \eval{v_1.f(\overline{v},e_i, \overline{e})}{v_1.f(\overline{v},e_i,'\overline{e})} }
}

\newcommand{\ecsmDt}{
\inferrule[E-Csm-Dt]
{ \eval {e_1} {e_1'} }
{ \eval {f(e_1')(\overline{e_2})}{f(e_1')(\overline{e_2})}}
}

\newcommand{\ecsmArg}{
\inferrule[E-Csm-Arg]
{ \eval {e_i} {e_i'} }
{ \eval {f(v_1)(\overline{v},e_i,\overline{e})} {f(v_1)(\overline{v},e_i,\overline{e})}}
}
\newcommand{\efun}{
\inferrule[E-Fun]
{ \overline{\eval e v} \\ \body(f) = (\overline{x},e) \\ \eval {[\overline{x} \mapsto \overline{v}]e} v}
{ \eval {f(\overline{e})} v}
}

\newcommand{\dtrGen}{
\inferrule[DtrGen]
{ \defs(C) = \classf C y T D {\overline{Fun}} \\ \fun x {e} \in \overline{Fun} }
{ \dtrBody(f,C) = (\overline{y},\overline{x},e)}
}

\newcommand{\dtrIt}{
\inferrule[DtrIt]
{ \defs(C) = \classf C y T D {\overline{Fun}} \\ \defs(D) = \interface D {Dtr} \\ \fun x e \in \overline{Dtr}}
{ \dtrBody(f,C) = (\varnothing,\overline{x},e)}
}

\newcommand{\csmCtr}{
\inferrule[CsmCtr]
{ \defs(f) = \csmcase e \\ \case {C(\overline{y})}{e}  \in \overline{\case P e}}
{ \csmBody(f,C) = (\overline{y},\overline{x},e) }
}

\newcommand{\csmDt}{
\inferrule[CsmDt]
{ \defs(f) = \csmcase e \\ \case {\_}{e}  \in \overline{\case P e}}
{ \csmBody(f,C) = (\varnothing,\overline{x},e) }
}

% inheritance is not subtyping
%TODO: we do not have casts/subtyping

\section{Formalization}\label{sec:formalization}
In this section, we first formalize the \emph{type-directed} bidirectional transformation between functional and object-oriented decomposition
in the \name calculus.
Our formalization combines ideas from Featherweight Java~\citep{igarashi2001featherweight} and \citet{popl2019}'s work.
We then concretely show how to transform \autoref{setoop} and \autoref{setfp} back and forth using \name.
% \paragraph*{Assumptions}
% For simplicity, we assume that pattern variable names are consistent with 
% those in their corresponding constructor.

\subsection{Syntax}
\autoref{syntax} gives the syntax of \name, which is a minimal calculus capturing essential features for functional and object-oriented
decomposition with a Scala-like syntax. 
To clarify the formalization, we use keywords different from Scala, where
Scala's \lstinline{trait} maps to \kwinterface; \lstinline{sealed trait} maps to \kwdata{}; and \lstinline{case class} maps to \kwcase.

A \name program ($L$) consists of some definitions ($Def$) followed by an expression ($e$).

% Unlike FJ, \name does not allow casts and subtyping, as classes are not used as types.

\begin{figure}[t]
  \begin{displaymath}
    \begin{array}{l}
      \begin{array}{llcl}
        \text{Program}
        & L & \Coloneqq & Def ; L \mid e\\
        \text{Definition}
        & Def & \Coloneqq & Dt \mid It \mid Ctr \mid Gen \mid Csm \\
        \text{Datatype}
         & Dt & \Coloneqq & \sealed D\\
        \text{Interface}
         & It & \Coloneqq & \interface D {Dtr}\\
        \text{Destructor}
         & Dtr & \Coloneqq & Dec \mid Fun\\
        \text{Constructor}
         & Ctr & \Coloneqq & \ctr C D\\
        \text{Generator}
        %  & gen & \Coloneqq & \class C {C(\overline{e})} {fun}\\
         & Gen & \Coloneqq & \class C D {Fun}\\
         \text{Consumer}
         & Csm & \Coloneqq & \csmcase e\\
        %  \text{Body}
        %  & body & \Coloneqq & \overline{P \Rightarrow e} \mid e\\
         \text{Declaration}
         & Dec & \Coloneqq & \dec x\\
         \text{Function}
         & Fun & \Coloneqq & Dec = e\\
        \text{Expression}
         & e & \Coloneqq & x %\mid \overline{(x: T)} \Rightarrow t 
          \mid e_1.f(\overline{e_2})  \mid f(e_1)(\overline{e_2}) \mid C(\overline{e}) \mid \newSeq{C}{e}\\
         \text{Pattern}
         & P & \Coloneqq & C(\overline{x}) \mid \_ \\
        %  \text{Modifier}
        %  & m & \Coloneqq & \kwsealed \mid \kwabstract \mid \kwcase\\
        \text{Types}
        & T & \Coloneqq & D \mid \overline{T} \rightarrow T\\
        & D & \Coloneqq & \text{datatype or interface names} \\
        & C & \Coloneqq & \text{constructor or generator names} \\
        & f & \Coloneqq & \text{consumer or destructor names} \\
        % \text{Names}
        % & D & datatype or interface name \\
      \end{array}
    \end{array}
  \end{displaymath}
  \caption{Syntax of \name}
  \label{syntax}
\end{figure}

% $\overline{f}$ denotes a possibly empty sequence.

\paragraph*{Definitions} A definition can be a datatype ($Dt$), an interface ($It$), a constructor ($Ctr$), a generator ($Gen$), or a consumer ($Csm$).
An $\interface D {Dtr}$ declares some destructors ($Dtr$), which can have an optional default implementation.
A $\class C D {Fun}$ implements the destructors declared by the interface $D$, no more and no less.
An instance of $\sealed D$ can be constructed via its constructors $\ctr C D$.
A consumer $\csmcase e$ has two parameter lists, where the first parameter has a reserved variable name $\kwself$ and must be of a datatype $D$.
The body of a consumer is some named patterns with an optional wildcard pattern in the end.

\paragraph*{Expressions} Metavariable $e$ ranges over expressions. Expressions include variables $x$, method selections $e_1.f(\overline{e_2})$, consumer applications $f(e_1)(\overline{e_2})$, constructor calls $C(\overline{e})$, or new expressions $\new{C}{\overline{e}}$.
We assume that the special variable $\kwthis$ is used only in destructors and the reserved variable $\kwself$ is only used for the datatype argument in consumers.

\paragraph*{Notations}
We write $\overline{\text{overline}}$ as a shorthand for a possibly empty sequence.
Semicolon denotes the concatenation of sequences.
Also, in the rest of the paper, we consider that variables with and without an overline are distinct, e.g. $x$ is distinct from $\overline{x}$. 
% \paragraph*{Types} 
% no subtyping

% We assume that the set of variables includes the special variable this, which cannot be used as the name of an argument to a method
% this must explicitly written for destructor calls on the object itself.

% \paragraph*{Typing}: why no subtyping etc.

\begin{figure}[t]
\begin{tabular}{ll}
\dt & Datatypes\\
\ct & Interfaces\\
$\ctrs{D}$ & Constructors of datatype $D$\\ 
$\dtr{D}$ & Destructors of interface $D$\\  
$\generator{D}$ & Generators of interface $D$\\
$\consumer{D}$ &  Consumers of datatype $D$\\
$\sig(N)$ & Signature of constructor/generator/consumer indexed by name $N$ \\
% $\csmBody(f,C)$ & Case clause for constructor $C$ in consumer $f$\\
% $\dtrBody(f,C)$ & Body of destructor $f$ in generator $C$\\
$\dtrType(f,D)$ & Signature of destructor $f$ for interface $D$\\
$\defs(N)$ & Definitions indexed by name $N$\\
% Caseof & $\caseof f C$ & = & if $\kwdef \ f(x : D)(\overline{x : T}): T = \overline{\kwcase\ C(\overline{x}) \Rightarrow e}\ \land$ \\
\end{tabular}
\caption{Preprocessed global context $\Delta$ for transformation}
\label{gctx}
\end{figure}

% \framebox{$\trans L T L'$}
\begin{figure}
\begin{mathpar}

\inferrule[It2Dt]
{ D \in \ct \\  \overline{\dtrToCsmCtx {\kwthis} {D} {Dtr} {Csm}} \\ \trans L T L'}
{ \trans {\interface{D}{Dtr}; L} T {\sealed{D}; \overline{Csm}; L'}}

\inferrule[Gen2Ctr]
{ D \in \ct \\ \trans L T L'} %  
{ \trans {\class {C} {D} {Fun}; L} T {\ctr{C}{D}}; L'}

\inferrule[Dec2Csm]
{ \overline{C} = \generator{D} \\ \overline{\dtrToCase{\defs(C)} {\case P e}}}
{ \dtrToCsm {\dec x} {\csm \overline{\case P e}}}

\inferrule[Fun2Csm]
{ \dtrToCsm {\dec x} {\csmcase e} \\ \transCtx {\overline{x:T}} e T e'}
{ \dtrToCsm {\dec x = e} \csmcase {e}; \kwcase\; \_ \Rightarrow [\kwthis \mapsto \kwself]e'}

% \framebox{$\dtrToCase {Gen} {\case P e}$}
\inferrule[Fun2Case] % reuse
{ \dec x = e \in \overline{Fun} \\ \transCtx {\overline{y:S}, \overline{x:T}} e T e'}
{  \dtrToCase {\classf {C} {y} {S} {D} {\overline{Fun}}} {\case {C(\overline{y})} {[\kwthis \mapsto \kwself]e'}}}
\end{mathpar}
% $\inbracket{\kwdef \ f(x : D)(\overline{x : T}): T  = \overline{\kwcase\ P \Rightarrow e}} \text{ when } f \in \bigcup\limits_{D_i \in \ct}\consumer{D_i} = \varnothing$ \\

% \framebox{$\trans e T e'$}
\begin{mathpar}
\selDtr

\newGen

\end{mathpar}

\begin{comment}

\inferrule[Var]
{ x : T \in \Gamma}
{ \trans x T x}

\inferrule[Select]
{ \trans {e_1} D {e_1'} \\ \overline{T} \rightarrow T = \dtrType(f,D) \\ \overline{\trans {e_2} T {e_2'}}}
{ \trans{e_1.f(\overline{e_2})}{T}{e_1'.f(\overline{e_2'})}}

\inferrule[Csm]
{\overline{\trans {e} T {e'}} \\ \trans L T L'}
{ \trans {\csmcase e; L} {T}  {\csmcase {e'};L'}}

\inferrule[New]
{ \sig(C) = \overline{T} \rightarrow D \\ \overline{\trans e T e'} } %\\ {\trans t {\overline{T} \arrow T} {t'}} }
{ \trans{\new C {e}}{D}{\new C{e'}}}
\end{comment}

\caption{Core translation rules for object-oriented decomposition}
\label{fig:trans-oo}
\end{figure}

\begin{figure}[t]
% \framebox{$\dtrToCsm {Dtr} {Csm}$}
\begin{mathpar}

\inferrule[Dt2It]
{ D \in \dt \\ \overline{f} = \consumer{D} \\ \overline{\Delta; \Gamma \vdash \defs(f) \leadsto Dtr}  \\ \trans L T L'}
{ \trans {\sealed D; L} T {\kwinterface \; D \{ \overline{Dtr}\} }; L'}
 
\inferrule[Ctr2Gen]
{ D \in \dt \\ \overline{f} = \consumer{D} \\ \overline{\Delta;\Gamma,\overline{x: T} \vdash_C {\defs(f)} \leadsto {Fun}}\\ \trans L T L'}
{ \trans {\ctr{C}{D}; L} T \classf{C}{x}{T}{D}{\overline{Fun}}; L'}
% { \textsf{fresh } x \\ \overline{Gen} = \generator{D} \\  }

\inferrule[CsmElim]
{ D \in \ct \\ \trans L T L'}
{ \trans {\csmcase e; L} {T} L'}
\end{mathpar}

% \framebox{$\Delta; \Gamma \vdash Csm \leadsto Dtr$}
\begin{mathpar}
% \inferrule[Csm2Dec]
% { {\kwcase\; \_ \Rightarrow e} \notin \overline{P}}
% {  \Gamma \vdash \csm P \leadsto \dec }

\inferrule[Csm2Fun]
{ {\kwcase\; \_ \Rightarrow e} \in \overline{\case P e} \\ \transCtx{\kwself: D,\overline{x:T}} e T e'}
{  \Delta; \Gamma \vdash \csmcase e  \leadsto \dec x = [\kwself \mapsto \kwthis]e'}

\inferrule[Csm2Dec]
% { {\kwcase\; \_ \Rightarrow e} \notin \overline{P}}
{}
{  \Delta; \Gamma \vdash \csmcase e \leadsto \dec x}
\end{mathpar}
%% \cristina{Do we need Fun2Case as well? I don't fully follow Case2Fun as opposed to Csm2Fun.}
%% \weixin{Rename Gen2Case as Fun2Case}
% \framebox{$\csmToDtr {Csm} {Fun}$}
\begin{mathpar}
\inferrule[Case2Fun]
{ \case {C(\overline{y})}{e}  \in \overline{\case P e} \\ \transCtx{\kwself: D,\overline{x:T}} e T e'}
{ \csmToDtr { \csmcase e} {\dec x = [\kwself \mapsto \kwthis]e'}}

\appCsm

\inferrule[Obj2New]
{ \sig(C) = \overline{T} \rightarrow D \\ C \in \ctrs{D} \\ \overline{\trans e T e'}}
{ \trans{C(\overline{e})}{D}{\newSeq C {e'}}}
\end{mathpar}
\caption{Core translation rules for functional decomposition}
\label{fig:trans-fp}
\end{figure}

\begin{figure}
\begin{mathpar}

\inferrule[It2It]
{ \overline{\transCtx {\kwthis: D} {Dtr} {T} {Dtr'}} \\ \trans L T L'}
{ \trans {\interface{D}{Dtr}; L} T {\interface{D}{Dtr'}}; L'}

\inferrule[Gen2Gen]
{ \overline{\transCtx{\kwthis:D, \overline{x:T}} {Fun} T {Fun'}}\\ \trans L T L'} %  
{ \trans {\class C D {Fun}; L} T {\class C D {Fun'}}; L'}
%% { \trans {\class C D {Fun}; L} T {\ctr{C}{D}}; L'}
% \inferrule[SubGenerator]
% { \exists E \; s.t. D \in \generator{E}}
% { \trans {\class C {D(\overline{e})} {M}} {OK} {\case C E}}

\inferrule[Sel2Sel]
{ \trans {e_1} D {e_1'} \\ \dtrType(f,D) = \overline{T} \rightarrow T\\ \overline{\trans {e_2} T {e_2'}}}
{ \trans{e_1.f(\overline{e_2})}{T}{e_1'.f(\overline{e_2'})}}

\inferrule[New2New]
{  \sig(C) = \overline{T} \rightarrow D \\ \overline{\trans e T e'}} %\\ {\trans t {\overline{T} \arrow T} {t'}} }
{ \trans{\newSeq C e}{D}{\newSeq C {e'}}}

\end{mathpar}
\caption{Some of the rules for types not selected for transformation}
\label{fig:skip}
\end{figure}

\subsection{Type-directed transformation}

% \paragraph*{Global context}
There is a preprocessing pass for collecting a global context $\Delta$ before the transformation actually happens. The information contained in $\Delta$ is given in \autoref{gctx}.
% Note that, in order to keep the formalization as simple as possible, we overload $\sig$ and $\defs$ to work with different types of arguments.
These definitions are trivial and, due to lack of space, are explained in plain text. 
%More formal defintions can be found in \autoref{sec:additional-formalization}. 
Intuitively, one can think of this preprocessing step as turning some of the
program's syntax into a logical encoding as sets and maps. As we shall see later, having
$\Delta$ precomputed will simplify the subsequent translation rules.

$\trans L T L'$ denotes the type-directed translation, which states that under the context of  $\Delta$ and $\Gamma$, a program $L$ of type $T$
is translated to $L'$.
Core translation rules are split into two sets respectively for object-oriented decomposition (\autoref{fig:trans-oo}) and
functional decomposition (\autoref{fig:trans-fp}).
Auxiliary translation rules may require additional information collected along the way.
For instance, $\vdash_D$ in \textsc{It2Dt} passes the interface name $D$ to the translation rules on destructors.
Similarly, we use $\vdash_f$ in \textsc{Fun2Case} to specialize the generator to case clauses translation for a given destructor $f$.
Lastly, $\vdash_{C}$ translates a specific case clause for constructor $C$ to a destructor in \textsc{Case2Fun}.  

Apart from the core rules, \autoref{fig:skip} gives some of the translation rules for skipping over code that is not meant to be translated
(Complete rules can be found in \autoref{sec:additional-formalization}).
All the rules in the figure work by only transforming the inner expressions. 
This is explained in more detail in \autoref{sec:given-types}.
Our main translation rules are explained by example in \autoref{sec:oo2fp}~and~\autoref{sec:fp2oo}. 
%Besides the core translation rules, we also provide rules that skip over code that the user does not want to translate. 
Notably, our translation rules make use of type information. 
\autoref{sec:why-type} provides concrete examples of why type information is essential. 

%% \section{\name by examples}\label{sec:example}
%% % To better understand the formalized transformation rules, 
%% In this section, we concretely show how to bidirectionally transform the integer sets example discussed in \autoref{sec:overview} step by step
%% using the rules given in \autoref{sec:formalization}.

\subsection{From object-oriented decomposition to functional decomposition}\label{sec:oo2fp}
Recall the code snippet shown in \autoref{setoop}.
After preprocessing, the global context contains:

\begin{tabular}{lcl}
\ct & = & \{ \texttt{Set} \}\\
\dtr{\texttt{Set}} & = & \{ \texttt{isEmpty, contains, insert, union} \}\\
\generator{\texttt{Set}} & = & \{ \texttt{Empty, Insert, Union} \} \\
\end{tabular}

The translation starts from the \lstinline{Set} interface. 
Since \lstinline{Set} is in \ct, it is translated to a datatype together with some consumers by the rule \textsc{It2Dt}.
% \textsc{It2Dt} transforms an interface into a datatype definition and its destructors into consumers.
As indicated by \dtr{\texttt{Set}}, \lstinline{isEmpty}, \lstinline{contains}, \lstinline{insert} and \lstinline{union} are destructors of \lstinline{Set}, 
which are translated to consumers under the term context that \lstinline{this} is of type \lstinline{Set}.
These destructors are treated differently depending on whether they have a default implementation. 
\lstinline{isEmpty} and \lstinline{contains} are declarations and hence are dealt by \textsc{Dec2Csm}. 
For instance, \lstinline{isEmpty} is augmented with the variable \lstinline{self} of type \lstinline{Set} and
its body is represented by case clauses collected from generators of \lstinline{Set}.
\textsc{Fun2Case} is applied to every generator of \generator{Set} (i.e. \lstinline{Empty}, \lstinline{Insert}, and \lstinline{Union}). 
\textsc{Fun2Case} finds out the definition of \lstinline{isEmpty}, translates the body and wraps it into a case clause. 
Particularly interesting is the \lstinline{Union} class, where \lstinline{isEmpty} is recursively called on its components.
With \lstinline{isEmpty} being a consumer afterwards, destructor selections on \lstinline{isEmpty} (e.g. \lstinline{s1.isEmpty}) should be translated into consumer calls (e.g. \lstinline{isEmpty(s1)}).
This is done by the rule \textsc{Sel2App} in a type-directed manner, which makes sure that \lstinline{s1} is of type \lstinline{Set} and \lstinline{isEmpty} is in \dtr{\lstinline{Set}}.
Therefore, the case clause for \lstinline{Union} is \lstinline{case Union(s1,s2) => isEmpty(s1) && isEmpty(s2)}.

On the other hand, destructors with a default implementation like \lstinline{insert} and \lstinline{union} are dealt by \textsc{Fun2Csm}. 
Building upon \textsc{Dec2Csm}, \textsc{Fun2Csm} additionally appends a wildcard clause constructed from the default implementation to the case clauses collected from generators.
Let us take a closer look at the \lstinline{insert} defined inside \lstinline{Set}. It calls \lstinline{new} on \lstinline{Insert}, 
which is in \generator{\lstinline{Set}} and will be a constructor after translation. Thus, \lstinline{new Insert(this,i)} is translated to
\lstinline{Insert(this,i)} by the rule \textsc{New2Ctr}. Also, the \lstinline{this} variable is substituted by the variable \lstinline{self}.

The translation then goes to generator definitions, \lstinline{Empty}, \lstinline{Insert}, and \lstinline{Union}. 
They are simply translated to constructors with their body dropped by the rule \textsc{Gen2Ctr}. 

Finally, the expression is translated, where \textsc{SelDtr} is repeatedly applied to convert destructor selections to consumer applications.

\subsection{From functional decomposition to object-oriented decomposition}\label{sec:fp2oo}
Preprocessing \autoref{setfp} gives us the following global context:

\begin{tabular}{lcl}
\dt & = & \{ \texttt{Set} \}\\
\ctrs{\texttt{Set}} & = & \{ \texttt{Empty, Insert, Union} \} \\
\consumer{\texttt{Set}} & = & \{ \texttt{isEmpty, contains, insert, union} \}\\
% \textsc{It} \textsc{Dtr} \textsc{Gen}& = & $\varnothing$\\
\end{tabular}

Now that \lstinline{Set} is a datatype in \dt, the rule \textsc{Dt2It} is applied, 
which translates a datatype into an interface. 
Consumers of \lstinline{Set}, revealed by \consumer{\texttt{Set}}, are translated to destructors
using either the rule \textsc{Csm2Dec} or \textsc{Csm2Fun} depending on whether they contain a wildcard pattern.
If a wildcard pattern is not contained (e.g. \lstinline{isEmpty} and \lstinline{contains}), the rule \textsc{Csm2Dec} is applied,
which returns the signature of the consumer with first parameter list (\lstinline{self: Set}) dropped. 
If a wildcard pattern exists\footnote{If the body of a consumer is an ordinary expression $e$, it is viewed as a syntactic sugar of $\kwcase\; \_ \Rightarrow e$} (e.g. \lstinline{insert} and \lstinline{union}), 
a default implementation is constructed from the right hand side of the wildcard pattern. 
The translation on \lstinline{insert} deserves some explanation.
Since \lstinline{Set} will become an interface after translation, 
consumer applications on \lstinline{Set} (e.g. \lstinline{contains(self)(i)}) are rewritten as destructor selections \lstinline{set.contains(i)} by the rule \textsc{App2Sel}. 
Also, constructor applications on generators of \lstinline{Set} (e.g. \lstinline{Insert(self,i)}) are replaced by new expressions (e.g. \lstinline{new Insert(self,i)}) using the rule \textsc{Obj2New}.
Moreover, references to the \lstinline{self} variable are substituted as \lstinline{this}.

The \textsc{Ctr2Gen} rule is then applied to every constructor of \lstinline{Set}, namely \lstinline{Empty}, \lstinline{Insert}, and \lstinline{Union}, for producing a generator.
For the case of \lstinline{Union}, the translation walks through the consumers of \lstinline{Set}. 
If there exists a case clause for \lstinline{Union}, a destructor will be produced by the rule \textsc{Case2Fun}.

Next goes to the consumer definitions.  Since \lstinline{insert}, \lstinline{isEmpty}, \lstinline{contains}, and \lstinline{union} are 
all in \consumer{\lstinline{Set}}, they are eliminated by the rule \textsc{CsmElim}.

Finally, \textsc{App2Sel} repeatedly rewrites consumer applications to destructor selections for the last expression.

% \lstinline{Set} are \lstinline{Empty}, \lstinline{Insert}, and \lstinline{Union}.

\subsection{Transforming selected types} \label{sec:given-types}
The programs discussed so far contain only one datatype or one interface.
However, multiple datatypes and interfaces may coexist in a complicated program.
For such programs, a user may want to select certain algebraic datatypes or interfaces rather than all of them for transformation.
Our transformation is also applicable in this scenario. The only additional step required is modifying the global context according to the selected types.

Let $\overline{D}$ be the types selected for transformation, then for every type $D$ in the program we do:
\[
D \notin \overline{D} \land D \in \dt \Rightarrow
\begin{cases}
\dt  = \dt / D\\ 
\ctrs{D} = \varnothing \\ 
\consumer{D} = \varnothing \\ 
\end{cases} \quad\quad
D \notin \overline{D} \land D \in \ct \Rightarrow
\begin{cases}
\ct  =  \ct / D\\ 
\generator{D} = \varnothing \\ 
\dtr{D} = \varnothing \\ 
\end{cases}
\]

That is, if a datatype $D \in \dt$ (interface $D\in\ct$) is not in the set of types selected for transformation $\overline{D}$,
then we remove it from \dt{} (\ct). We also set the corresponding
set of constructors and consumers (generators and destructors) to $\varnothing$.
Basically, \dt{} or \ct~ will only contain the types selected for transformation. 
Some of the rules for skipping over code that is not meant to be translated (by only translating the inner expressions) are given in \autoref{fig:skip}. % in \autoref{sec:additional-formalization}.

\paragraph*{Transforming selected interfaces} 
To see concretely how the transformation works only on selected types, suppose the following class hierarchy for integer lists is defined together with \autoref{setoop}:

\begin{lstlisting}[linerange=28-34]
trait List { def contains(i: Int): Boolean }
object Nil extends List {
  def contains(i: Int) = false
}
class Cons(n: Int, xs: List) extends List {
  def contains(i: Int) = i == n || xs.contains(i)
}
\end{lstlisting}%APPLY:LIST_OOP
We would only like to transform \lstinline{Set}.
Since \lstinline{List} is not one of the selected types, 
it is removed from \ct{} and \generator{\lstinline{List}}/\dtr{\lstinline{List}} are deleted,
resulting in a $\Delta$ similar to that in \autoref{sec:oo2fp}.
Consequently, different sets of rules will be applied to the \lstinline{Set} and \lstinline{List} hierarchies.
The rule \textsc{It2It} is applied for \lstinline{List} and \textsc{Gen2Gen} is applied to \lstinline{Nil} and \lstinline{Cons}, which translate their inner expressions only. 
Moreover, destructor selections on \lstinline{List} objects and new on \lstinline{Cons} are preserved by the rule \textsc{Sel2Sel} and \textsc{New2New}. Note that the translation rules 
\textsc{It2It}, \textsc{Gen2Gen}, \textsc{Sel2Sel} and \textsc{New2New} can be found in \autoref{fig:skip}. % in \autoref{sec:additional-formalization}.

\paragraph*{Transforming selected datatypes} 
Oppositely, suppose \lstinline{List} is implemented in a functional way together with sets shown in \autoref{setfp}:
\begin{lstlisting}[linerange=28-30]
sealed trait List
case object Nil extends List
case class Cons(n: Int, xs: List) extends List
\end{lstlisting}%APPLY:LIST_FP
And \lstinline{Set} is the only candidate for transformation.
Then we get a $\Delta$ similar to that in \autoref{sec:fp2oo}.
Unlike \lstinline{Set}, \lstinline{List} is processed by the rule \textsc{Dt2Dt} since it is 
not in \dt{}. Therefore, \lstinline{List} is still a datatype after translation. 
\lstinline{Nil} and \lstinline{Cons} are also unchanged by the rule \textsc{Ctr2Ctr} and 
\lstinline{contains} remains a consumer on \lstinline{List} by the rule \textsc{Csm2Csm}.
Moreover, constructor calls on \lstinline{Cons} and application of \lstinline{contains} on \lstinline{List}s are
retained by the rule \textsc{Obj2Obj} and \textsc{App2App}.

\subsection{Why the transformation should be type-directed}\label{sec:why-type}
Importantly, the \lstinline{List} example illustrates why the transformation should be type-directed.
Notice that \lstinline{contains} is defined on both \lstinline{Set} and \lstinline{List}.
For the case when \lstinline{contains} is a destructor, 
expressions of the form \lstinline{x.contains(i)} should be transformed differently according to the type of \lstinline{x}.
If the type of \lstinline{x} is \lstinline{Set} which is in \ct, the rule \textsc{Sel2App} is applied, which transforms the expression to \lstinline{contains(x)(i)}. 
Otherwise, the rule \textsc{Sel2Sel} is applied and \lstinline{x.contains(i)} will be the same after transformation.  
Likewise, for the case when \lstinline{contains} is an overloaded consumer on both \lstinline{Set} and \lstinline{List}, 
expressions of the form \lstinline{contains(x)(i)} should be distinguished by the type of \lstinline{x}.
When \lstinline{x} is of type \lstinline{Set}, the rule \textsc{App2Sel} transforms the expression to \lstinline{x.contains(i)}  
due to the fact that \lstinline{Set} is in \dt. In contrast, the rule \textsc{App2App} gives back \lstinline{contains(x)(i)} when \lstinline{x} 
is a \lstinline{List} object. 

However, without being guided by types, a syntax-directed approach such as the one 
used by ~\citet{popl2019} would transform expressions of the same form uniformly, resulting in an erroneous program after transformation.

% overloading

% Collection example
% $$
% \begin{array}{lcl}
% \dt' & = & \ct \\ 
% \ct' & = & \dt \\ 
% \generator{D}' & = & \ctrs{D} \\ 
% \ctrs{D}' & = & \generator{D} \\ 
% \dtr{D}' & = & \consumer{D} \\ 
% \consumer{D}' & = & \dtr{D} \\ 
% \end{array}
% $$

\section{Soundness results} \label{sec:soundness}
In this section, we discuss our main soundness results. 
We start by giving 
%necessary machinery for our soundness proofs. In particular,
%we discuss well-formedness (Section~\ref{sec:well-formedness}) and
the operational semantics of \name in ~\autoref{sec:big-step}.
Then, we proceed to provide
the soundness theorems in~\autoref{sec:metatheory}.

%% \subsection{Well-formedness}\label{sec:well-formedness}

%% The well-formedness judgment checks that: 1) variables are in scope;
%% 2) no additional methods are defined on classes.
%% 3) patterns are exhaustive.

% nominal typing
\begin{figure}[t]
\begin{displaymath}
\begin{array}{lcl}

v & ::= & obj(C,\overline{v})

%v & ::= & \obj{v} \mid \newSeq{C}{v}

% v,u & ::= & \new{C}{v} \mid C(\overline{v})
\end{array}
\end{displaymath}

\framebox{Evaluation contexts}
\begin{displaymath}
\begin{array}{lcl}
E & ::= & \square \mid C(\overline{v},\square,\overline{e}) \mid \new{C}{\overline{v},\square,\overline{e}} \mid f(\square)(\overline{e}) \mid f(v)(\overline{v},\square,\overline{e}) \mid \square.f(\overline{e}) \mid v.f(\overline{v},\square,\overline{e})
\end{array}
\end{displaymath}

\framebox{$\eval e e'$}
\begin{mathpar}

  \eCongr

  \ectr

  \enew

  \ecsmInvk

\edtrInvk
%
  % \edtrRecv
%
  % \edtrArg
%
  % \ecsmDt
%
  % \ecsmArg
\end{mathpar}
\caption{Small-step semantics of \name}
\label{semantics}
% TODO: small-step semantics and type-soundness proof

%Evaluation contexts

% FJ: if a term is well typed and it reduces to a normal form, then it is either a value or an expression that gets stuck at a downcast.

% Preservation: if Gamma |- e: C and e -> e', then Gamma |- e': C.

% Progress: Suppose e is a well-typed expression
% (1) if e includes new C(e) as a 

% \framebox{$\eval e {e'}$}
% \begin{mathpar}
% \inferrule[E-InvkNew]
% { \dtrBody(f,C(\overline{v})) = (\overline{x},e) }
% { \eval{(\new{C}{v}).f(\overline{u})}{[\overline{x} \mapsto \overline{u},\kwthis \mapsto \new{C}{v}]e}}

% \inferrule[E-Csm]
% { \csmBody(f,C(\overline{v})) = (\overline{x},e) }
% { \eval {f(C(\overline{v}))(\overline{u})} {[\overline{x} \mapsto \overline{u}]e}}

% \inferrule[E-Fun]
% { \body(f) = (\overline{x},e) }
% { \eval {f(\overline{v})} {[\overline{x} \mapsto \overline{v}]e}}
% \end{mathpar}
% \caption{Small-semantics of \name}
% \label{semantics}
\end{figure}

\begin{figure}[t]
\begin{mathpar}
\framebox{$\dtrBody(f,C) = (\overline{y},\overline{x},e)$}\\
\dtrGen

\dtrIt
\end{mathpar}

\framebox{$\csmBody(f,C) = (\overline{y},\overline{x},e)$}\\

\begin{mathpar}
\csmCtr

\csmDt

% \framebox{$\body(f) = (\overline{x},e)$}\\
% \inferrule*
% {\defs(f) = \fun e}
% { \body(f) = (\overline{x},e)}
\end{mathpar}
\caption{Destructor/consumer lookup}
\label{fig:compute-delta}
\end{figure}
%\weixin{Should it be small-step or big-step for the proof purpose?}

\subsection{Small-step semantics}\label{sec:big-step}

We introduce a call-by-value small-step semantics for \name, which will be used when proving the soundness of the
translation rules. The rules for our evaluation relation
$\rightarrow$ are given in \autoref{semantics}.

%next subexpression to be reduced.
An evaluation context is an expression with a hole $\square$, where the hole denotes the next expression to be evaluated.
$E[e]$ replaces the hole with the expression.
We assume that $obj(C,\overline{v})$ denotes an object created from a generator (\textsc{E-New}) or a constructor (\textsc{E-Ctr}) $C$, where
$\overline{v}$ are the values of the corresponding fields.
Rule \textsc{E-Congr} applies one evaluation step to expression $e_1$ in the evaluation context $E[e_1]$, thus evaluating $E[e_1]$ to $E[e_2]$.
For the \textsc{E-Dtr} rule, we bound the receiver object 
to \texttt{this}, and all the fields of $f$ to their corresponding values, in order to evaluate its body $e$.
The rule \textsc{E-Csm} is very similar to \textsc{E-Dtr} with the difference that,
instead of binding the object $obj(C,\overline{v})$ to \texttt{this},
we bind it to the datatype argument $\kwself$.
Note that before \textsc{E-Dtr} and \textsc{E-Csm} apply, the repeated application of \textsc{E-Congr} 
evaluated all the sub-expressions to values.
\textsc{E-Dtr} and \textsc{E-Csm} rely on two auxiliary definitions given in \autoref{fig:compute-delta}.
For $\dtrBody$, we compute a tuple containing the class fields, destructor arguments and destructor body by looking up destructor $f$ defined on generator $C$.
Similarly, for $\csmBody$ the elements of the resulting tuple are: constructor field names, parameter names in the second list, and the right-hand side of the case clause $C$ on consumer $f$. We assume the usual rules for substitution.

%\autoref{semantics} gives the call-by-value big-step semantics of \name.

\subsection{Soundness theorems} \label{sec:metatheory}
% well-typedness of the derived expression Correctness of transformation
%\cristina{Added some wording.}

Theorems~\ref{theorem:preservation} and~\ref{theorem:progress} state the type safety of \name.
Then, \autoref{theorem:tp} and \autoref{sp} state the soundness of our translation. In particular, \autoref{theorem:tp} captures the type and syntax preservation of the translation, and \autoref{sp} states the preservation of semantics.
Additional lemmas and proofs are left in \autoref{sec:proofs}.  We also have a discussion on well-formedness.

\paragraph{Type safety of \name} Given that our
typing rules are integrated with the type-directed translation rules in Figure~\ref{fig:trans-oo} and Figure~\ref{fig:trans-fp},
the type safety theorems ignore the result of the actual translation (by using $\_$), and only focus on type derivations. 

\begin{theorem}[Preservation]\label{theorem:preservation} 
  Suppose $e$ is a well-typed expression.
  if $\Delta; \boldsymbol{\cdot} \vdash e : T \leadsto \_$ and $\eval e {e'}$, then $\Delta; \boldsymbol{\cdot} \vdash e' : T \leadsto \_$.
\end{theorem}

\begin{theorem}[Progress]\label{theorem:progress} 
  If $\Delta; \boldsymbol{\cdot} \vdash e : T \leadsto \_$,
  %is a well-typed expression, 
  % if $e$ includes $\new C e.f(d)$ as a subexpression, then dtrBody(m,C) = (ys,xs,e0) and 
  then either $\eval e {e'}$ or $e$ is a value.
  % $\trans e T e'$ and $e \$, then $\trans {e_1} T {e''}$
\end{theorem}

\paragraph{Soundness of the translation}
%We now state our main soundness theorems.
\autoref{theorem:tp} states that,
if a well-typed program $P$ is translated to $P'$ under the typing environment $\Gamma$ and
global context $\Delta$, then $P'$ can be translated back to $P$ under the same
typing environment and the translated global context $\Delta'$. Intuitively,
transforming a \name program twice will give us back the original program, thus being
syntax preserving. Also, $P'$ has the same type $T$ as $P$ under the typing environment $\Gamma$ and global context $\Delta'$, thus being type preserving too.
%% \begin{theorem}[Well-formedness preservation]\label{theorem:wfp} if $\Gamma \vdash P$ and $\trans P T P'$, then $\Gamma \vdash P'$.
%% \end{theorem}

% \begin{theorem}[Soundness]\label{theorem:soundness} 
%   If $\trans e T e'$ and $\eval e {e''}$, then $e''$ is a value.
% \end{theorem}

\begin{theorem}[Transformation Preservation]\label{theorem:tp} if $\trans P T P'$ and $\Delta \leadsto \Delta'$, then $\transGeneric {\Delta'; \Gamma} {P'} T P$. 
\end{theorem}

%% \todo{We still don't handle the non-termination case. The progress theorem shows that programs in \name can't get stuck, although they can still be non-terminating.
%% We can think about saying something about the fact that if a program P is non-terminating, it will also be non-terminating after translation.}

\autoref{sp} states that if the evaluation of $P$ terminates with result $v$, then the
evaluation of $P'$ obtained through the translation also terminates with result $v$. Alternatively, if the execution of $P$ doesn't terminate,
neither does the execution of $P'$.

\begin{theorem}[Semantics preservation]\label{sp} Given $\trans P T P'$, the following two conditions must hold: (1) if $ P \longrightarrow^* v$, 
  then ${P'} \longrightarrow^* v$, and (2) if $P$ diverges, then so does $P'$.
\end{theorem}

\paragraph*{Well-formedness}
Another property of interest for our transformation besides type and semantics preservation,
is the preservation of well-formedness. 
Informally, a \name program is well-formed if variables are in scope, a class implements the exact set of methods declared by the interface, 
and patterns are exhaustive. Intuitively, it can be seen that proving preservation of well-formedness is straightforward. 
We leave the formal proof as future work.

% \begin{theorem}[Roundtrip / Syntax preservation] if $\trans P T P'$ then $\trans {P'} T P$. 
% \end{theorem}

%\input{sections/Example}
%TODO: A table summarizes all the programs?
\section{Implementation and Case studies}\label{sec:case}
Based on \name, we develop a bidirectional transformation tool called \systemname in Scala and conduct several case studies to demonstrate its applicability and effectiveness.
In this section, we discuss the implementation and some of the most interesting case studies. 

\subsection{Implementation and limitations of \systemname}\label{sec:impl}
\systemname employs the \emph{Scalameta} library \cite{scalameta}
for analyzing, transforming, and pretty-printing Scala programs.
Given a Scala program, \systemname gets its AST from the parser provided by Scalameta.
Following the formalization in \autoref{sec:formalization}, the AST is processed by two major passes. 
The first pass collects the global context $\Delta$ and the second pass does the type-directed transformation.
To accept more Scala programs, \systemname has a few more passes in-between such as making the use of \lstinline{this} explicit for \lstinline{class}es or \lstinline{trait}s.
After the transformation, we obtain a new AST. 
Finally, we pretty-print the new AST to produce the transformed program. 

% \paragraph*{Language constructs beyond \name}
Compared to \name, Scala is so rich in terms of language constructs that \systemname cannot handle all of them.
In order to transform programs used in our case studies, \systemname does extend \name to
handle \lstinline{if}-expressions, ordinary \lstinline{def} functions, \lstinline{val} declarations, etc.
Some constructs like \lstinline{if-} expressions have straightforward translation rules that just apply the translation to subexpressions like what \autoref{fig:skip} shows.
Some constructs may affect contexts such as \lstinline{def} and \lstinline{val} declarations, where contexts are changed accordingly before the translation goes on 
to the remaining program. Still there are many features that \systemname cannot handle, which will be discussed in \autoref{sec:morefeatures}.

As a prototype implementation, \systemname has certain limitations.
One limitation is that the soundness of translation on programs that use features beyond \name is not guaranteed, 
although the Scala compiler and test suites can be used to check the correctness of the translated programs.
Another limitation is that when applying \systemname to existing code that is not written in \name style, one has to adjust the code 
manually. Typical adaptions are moving the datatype parameter to first, fixing name inconsistencies between case class and patterns, 
rewriting nested patterns into top-level patterns or adding explicit type annotations for \systemname that can essentially be inferred by the Scala compiler.
These manual adaptions could potentially be automated by \systemname.

\subsection{Case studies overview}
% The case studies are collected from textbook or related work 
% to examine functional decomposition, object-oriented decomposition and/or both.
 % that is written in a pure object-oriented 
\autoref{tbl:study} summarizes the case studies and examples (separated by the line) with the associated file name, the description, the style, SLOC and translation time displayed for each case study.
The case studies are interesting because
\begin{itemize}
  \item \textbf{They are written in different styles}, covering object-oriented, functional, or mix styles;
  \item \textbf{They are extended in various ways}, including data variant and operation extensions;
  \item \textbf{They are non-trivial}, e.g. the SQL query processor \citep{rompf2012lightweight} and the prettier printer \citep{wadler2003prettier} are used as core of industrial-strength databases or compilers.
\end{itemize}
Besides larger case studies, there are also smaller examples ported from different textbooks for introducing functional or object-oriented programming~\cite{felleisen2010design, odersky2008programming} 
or literatures for data structure implementations \cite{okasaki1999purely,fredman1986pairing}.
% These extra examples increase the diversity of our case studies.
Overall, these case studies are rather comprehensive in examining different aspects of \systemname.
The remaining of this section will focus on the most interesting ones. 

% \texttt{List} & List \cite{okasaki1999purely}& FP  & 48 &\\
% \texttt{Bst} & Binary search tree \cite{okasaki1999purely} & FP & 49 &\\
% \texttt{AATree} & Efficient tree & FP & 43 &\\
% \texttt{Heap} & Simple heap & FP & 43 &\\
% With minor tweaks, we show that they can be automatically transformed into the other style using the Scala implementation of \name.
% We roughly measure the effort of making existing code fit into the \name style by counting 
% the changed SLOC.

\begin{table} \footnotesize
\caption{Summary of case studies.}
\begin{tabular}{lccccll}
  \toprule
\textbf{File name} & \textbf{Description} & \textbf{Style}& \textbf{SLOC} & \textbf{Time(ms)} \\
\hline
% LightFP & Traffic light & functional & 21& \\
% LightOOP & Traffic light &  object-oriented & 22&\\
\texttt{SqlFP} & SQL query processor \citep{rompf2019sql} & FP & 176 & 54.1\\
% \texttt{SqlOOP} & SQL query processor \cite{shallow} & OOP & 176 \\
\texttt{PrettierPrinter} & Pretty printer with alternative layouts \citep{wadler2003prettier} & FP & 195 & 33.2\\
\texttt{BoolNormalizer} & Boolean formula normalizer \citep{popl2019} & Mixed & 92 & 14.8\\
\texttt{ProgInScala} & Rendering library \& expression formatter \citep{odersky2008programming} & Mixed & 106 & 26.9\\
\texttt{Json} & JSON formatter and validator  & Mixed & 90 & 19.7 \\
\texttt{Scans} & Circuit DSL \cite{shallow} & FP & 63 & 10.7 \\
\hline
\texttt{RiverSystem} & River location \cite{felleisen2010design} & OOP & 23 & 3.9\\
\texttt{Shape} & Shape representation \cite{felleisen2010design} & OOP & 31 & 2.5\\
\texttt{Stack} & Stack implementation \cite{felleisen2010design} & OOP & 18 & 4.4\\
% \texttt{PaymentMethod} & Payment handler & OOP & 56 & 4.8\\
\texttt{PairingHeap} & Efficient heap implementation \cite{fredman1986pairing}& FP  & 29 & 5.0\\
\texttt{ListFP} & List implementation \cite{okasaki1999purely}& FP  & 38 & 3.9\\
\texttt{Bst} & Binary search tree implemetation \cite{okasaki1999purely}& FP  & 40 & 4.9\\
  \bottomrule
\end{tabular}
\label{tbl:study}
\end{table}

\subsection{SQL Processor}
\label{sec:sql}
Starting from an interpreter, \citet{rompf2019sql} show how to turn the simple interpreter into an efficient compiler through staging \`a la LMS~\cite{rompf2012lightweight}.
Their SQL processor implementation uses functional decomposition in Scala. 

% to facilitate adding operations that compile SQL queries
% to efficient Scala or C code. \citet{shallow} reimplement the SQL processor using object-oriented decomposition
% and argue that object-oriented decomposition brings extra modularity such as the ability to modularly add new variants.
% In this case study, we adapt these two implementations and show that they are interchangeable with \name.

% \paragraph*{Functional decomposition}
% A parser is implemented that parses a SQL query into a \lstinline{Operator}.
A SQL query is parsed into a relational algebra operator.
Initially, the \lstinline{Operator} datatype has 5 constructors: 
\lstinline{Scan}, \lstinline{Project}, \lstinline{Filter}, \lstinline{Join} and \lstinline{Print}.
% The \lstinline{Filter} constructor depends on the \lstinline{Predicate} datatype for
% comparing the (in)equality of two \lstinline{Ref} instances.
The core consumer on \lstinline{Operator} is \lstinline{execOp}, which accumulates the action
to the record in the parameter \lstinline{yld} according to the \lstinline{Operator}. 
The following excerpt is copied from \citet{rompf2019sql}'s paper:
\begin{lstlisting}
def execOp(o: Operator)(yld: Record => Unit): Unit = o match {
  case Scan(filename) => processCSV(filename)(yld)
  case Print(parent) => execOp(parent) { rec => printFields(rec.fields) }
  case Filter(pred, parent) => 
    execOp(parent) { rec => if (evalPred(pred)(rec)) yld(rec) }
  ...
}
\end{lstlisting}
% For example, a SQL query that finds out the room and title of all talks at 9 am
% \begin{lstlisting}
% select room, title from talks.csv where time = '09:00 AM'
% \end{lstlisting}
% is parsed into the following relational algebra operator:
% \begin{lstlisting}
% Project(Schema("room", "title"),
%   Filter(Eq(Field("time"),Value("09:00 AM")),
%     Scan("talks.csv")))
% \end{lstlisting}
% where the \lstinline{select}, \lstinline{from}, and \lstinline{where} clauses are respectively 
% parsed as \lstinline{Project}, \lstinline{Scan}, and \lstinline{Filter} operator.
% The condition used in the \lstinline{where} clause is captured by \lstinline{Eq} (a constructor of \lstinline{Predicate}) with its two operands being
% \lstinline{Field} and \lstinline{Value} (both are constructors of \lstinline{Ref}). 

% \lstinputlisting[linerange=9-17]{../app/src/main/scala/SqlFP.scala}%APPLY:OPERATOR

\noindent which is already written in a style close to \systemname. 
In fact, their entire implementation is written in this manner.
Consequently, a few renaming on variables is sufficient to fit their implementation into \systemname. 

Later on, \citet{rompf2019sql} introduce two new operators, \lstinline{Group} and \lstinline{HashJoin}, for
supporting \lstinline{group by} syntax and a more efficient hash-join operator.
With the functional decomposition version, these extensions have to be added by defining new case classes of \lstinline{Operator} and modifying
existing consumers on \lstinline{Operator} scattered around the file.
In contrast, it can be simplified by switching to object-oriented style and adding the extensions modularly as \lstinline{class}es with \systemname: 
\begin{lstlisting}
class Group(keys: Schema, agg: Schema, op: Operator) extends Operator {
  def resultSchema = keys ++ agg
  def execOp(yld: Record => Unit) = ...
}
class HashJoin(left: Operator, right: Operator) extends Operator {
  def resultSchema = resultSchema(left) ++ resultSchema(right)
  def execOp(yld: Record => Unit): Unit = ...
}
\end{lstlisting}
The transformed object-oriented code looks very much like the core part of \citet{shallow}'s hand-written implementation. 
Similarly, when the need of new operations such as compiling SQL queries to C backend arises, 
we can transform the implementation back to functional style to facilitate operation extensions.
Again, the \systemname-transformed functional version with extensions is almost identical to the manually extended version done by \citet{rompf2019sql}.

In fact, as acknowledged by \citet{rompf2019sql}, the decomposition style switch happens when
they grow this simple SQL processor into a realistic SQL engine. The decomposition style is switched
from functional to object-oriented to support a large number of operators such as various 
join operators (semi joins, anti joins, outer joins, etc.).
With \systemname, however, this switch can be done seamlessly without the pain of pervasive refactoring on the original code.

\subsection{Prettier printer}\label{sec:pprint}

\citet{wadler2003prettier} presents a pretty printer library that allows a document to be printed with different layouts efficiently.
The implementation was originally written in Haskell using functional decomposition.
\citet{wadler2003prettier} shows a simpler and less efficient version first.
A document is represented by a datatype (\lstinline{TDoc}) with 3 constructors for creating empty documents (\lstinline{TNil}), lifting strings into documents (\lstinline{TText}) and inserting line breaks (\lstinline{TLine}).
Other combinators on \lstinline{TDoc}  for adding indentations (\lstinline{nest}) and concatenating documents (\lstinline{<>}) are defined as consumers on \lstinline{TDoc},
along with a \lstinline{layout} consumer for printing a \lstinline{TDoc} as a string.

% with \lstinline{NIL}, \lstinline{TEXT}, \lstinline{LINE}, \lstinline{NEST} and \lstinline{:<>} as initial constructors 
The simple version is then extended to support alternative layouts.
\lstinline{TDoc} is extended with both new consumers (e.g. \lstinline{group}, \lstinline{flatten} and \lstinline{pretty}) as well as a new constructor \lstinline{Union}.
To further improve efficiency, another document datatype \lstinline{DOC} is defined and several consumers are introduced for compiling \lstinline{DOC} to \lstinline{TDoc} efficiently.

We ported the Haskell implementation into Scala and then adapted the ported implementation into \systemname style.
Most of the adaptions are trivial. Only two consumers deserve some attention since they use advanced forms of pattern matching. 

The first is the \lstinline{fits} consumer which uses a guarded wildcard pattern in the beginning:
\begin{lstlisting}
def fits(w: Int, x: TDoc): Boolean = x match {
  case _ if w < 0  => false
  case TText(s, x) => fits(w - s.length, x)
  case TNil        => true
  case TLine(i, x) => true
}
\end{lstlisting}  
We split \lstinline{fits} into two consumers and replace the guarded wildcard pattern as an \lstinline{if-expression}:
% calls the auxiliary consumer \lstinline{fitsAux} when 
\begin{lstlisting}
def fits(self: TDoc)(w: Int): Boolean = if (w < 0) false else fitsAux(self)(w)
def fitsAux(self: TDoc)(w: Int): Boolean = self match {
  case TText(s, x) => fits(x)(w - s.length)
  case TNil        => true
  case TLine(i, x) => true
}
\end{lstlisting}

The other one is \lstinline{be} function, which deeply pattern matches on a list of pairs.
\begin{lstlisting}
def be(w: Int, k: Int, xs: List[(Int, DOC)]): TDoc = xs match {
  case List()               => TNil
  case (i, NIL) :: z        => be(w, k, z)
  case (i, :<>(x, y)) :: z  => be(w, k, (i, x) :: (i, y) :: z)
  ...
}
\end{lstlisting}
Similar to \lstinline{fits}, we break \lstinline{be} into two mutually recursive functions with
nested patterns rewritten as top-level patterns:
% and change the signature of \lstinline{be} a.
\begin{lstlisting}
def be(w: Int, k: Int, xs: List[(Int, DOC)]): TDoc = xs match {
  case List()      => TNil
  case (n, x) :: z => beAux(x: DOC)(w, k, n, z)
}
def beAux(self: DOC)(w: Int, k: Int, n: Int, z: List[(Int, DOC)]): TDoc = self match {
  case NIL        => be(w, k, z)
  case :<>(x, y)  => be(w, k, (n, x) :: (n, y) :: z)
  ...
}
\end{lstlisting}

With these adaptions done, we obtained an object-oriented implementation of Wadler's printer transformed by \systemname.  
Here is an excerpt of the \lstinline{DOC} interface produced by \systemname:
\begin{lstlisting}
trait DOC {
  def <>(y: DOC): DOC = new :<>(this, y)
  def nest(i: Int): DOC = new NEST(i, this)
  def group: DOC = new :<|>(this.flatten, this)
  def flatten: DOC
  def pretty(w: Int): String = this.best(w, 0).layout
  ...
}
\end{lstlisting}
A simple example to demonstrate the transformed library could be:
\begin{lstlisting}
val d: DOC = text("[") <>              //[
  (line <> text("Hello") <>            //  Hello 
  line <> text("world")).nest(2) <>    //  world
  line <> text("]")                    //]
\end{lstlisting}
where the result of calling \lstinline{d.pretty(10)} is shown in comments.
Larger examples contained in the original implementation like trees and XMLs all work well in the transformed library.

Similar to the extensions discussed in \autoref{sec:sql}, 
adding extensions mentioned above can also be simplified by switching decomposition style back and forth with the help of \systemname.

% % Our second case study
% \scans~\cite{hinze2004algebra} is a domain-specific language for describing parallel prefix circuits.
% Gibbons and Wu \cite{gibbons2014folding} presents a Haskell implementation using shallow embeddings.
% Zhang and Oliveira~\cite{shallow} reimplement \scans in Scala. 
% We adapt their implementation 
% \scans originally consists

%Since the correctness of the translation has been proved by \autoref{sec:metatheory}, our evaluation focuses on other aspects of \name. 
%Specifically, we evaluate \name by answering the following questions:

\subsection{Boolean formula normalizer}\label{sec:case2}
The boolean formula normalizer case study is reproduced from \citet{popl2019}'s work. 
% which was originally implemented by \citet{danvy2011}.
% "inter-derive reduction-based and reduction-free negational normalization functions"
The goal is to develop a normalizer that evaluates a boolean formula to its negation normal form 
by repeatedly applying De Morgan's laws or double negation elimination rules.
The development is done through several iterations by switching the decomposition style on a particular type
back and forth for extensions.
This case study further examines \systemname's ability to transform selected types.
% The implementation is written using a mixed style
% Specifically, the following rewriting rules are applied:

% \begin{align*}
% \lnot(P \land Q) & = \lnot P \lor \lnot Q \\
% \lnot(P \lor  Q) & = \lnot P \land  \lnot Q \\
% \lnot\lnot P & = P
% \end{align*}
The implementation contains several datatypes for modeling boolean formulas (\lstinline{Expr}), redexes (\lstinline{Redex}), values (\lstinline{Value}) and 
found value/redex (\lstinline{Found}).
% consumers for conversions between these datatypes (\lstinline{asExpr} and \lstinline{eval}).
% The \lstinline{Found} datatype represents the search result for a redex.
Core consumers are \lstinline{search}, \lstinline{searchPos} and \lstinline{searchNeg}, which search for a redex in a \lstinline{Expr} using the continuation-passing style.
\begin{lstlisting}
trait Value2Found { def apply(value: Value): Found }

def search(self: Expr): Found = searchPos(self)(EmptyCnt)
def searchPos(self: Expr)(cnt: Value2Found): Found = self match {
  case EVar(n) => cnt.apply(ValPosVar(n))
  case ENot(e) => searchNeg(e)(cnt)
  case EAnd(l,r) => searchPos(l)(new AndCnt1(r,cnt))
  case EOr(l,r) => searchPos(l)(new OrCnt1(r,cnt))
}
\end{lstlisting}
Initially, \lstinline{Value2Found} was modeled as an interface with only destructor \lstinline{apply}.
Some instances of \lstinline{Value2Found} used in the definition above are given below:
\begin{lstlisting}[linerange=61-66]
object EmptyCnt extends Value2Found {
  def apply(value: Value): Found = FoundValue(value)
}
class AndCnt1(e: Expr, cnt: Value2Found) extends Value2Found {
  def apply(value: Value): Found = searchPos(e)(new AndCnt2(value,cnt))
} 
\end{lstlisting}%APPLY:CNT
We adapt local comatches (similar to anonymous objects) used by \citet{popl2019} to ordinary generators in \systemname. 

% The implementation by ~\citet{popl2019} model them as local comatches as they are used only once. 
% Local comatches are similar to anonymous objects, which will break the symmetricity. 
% \citet{popl2019} address this problem by designing a special syntax for explicitly providing a name and capturing the surrounding environment explicitly.
% This can hardly be enforced for an existing language like Scala.
% Our workaround is to define them as ordinary generators in \name.

In the second iteration, \lstinline{Value2Found} was switched to a datatype, renamed as \lstinline{Context} and extended
with a new consumer \lstinline{findNext}.
In the third iteration, \lstinline{Context} datatype was extended with another consumer \lstinline{substitute} and the \lstinline{Expr} datatype was extended with an \lstinline{evaluate} consumer. 
By switching \lstinline{Context} back to an interface, we derived the final implementation.

Unsurprisingly, \systemname makes the iterative development of the boolean formula normalizer smoothly by allowing style switches on selected types.

\subsection{Findings}
We have already established the semantic preserving nature of the translation in \autoref{sec:metatheory}. The case studies 
above additionally allow us to conclude a number of findings below on the effectiveness and applicability of \systemname. 
% \paragraph{Q1}
% The transformed programs all pass type-checking of the Scala compiler. 
% We additionally have some unit tests to check that the transformed programs produce the same result as the original programs.
% We further check the correctness by transforming the same program multiple times to see whether we can get back the original program.

\paragraph{\systemname facilitates extensions.} As we can see in \autoref{sec:sql}, the translated object-oriented code behaves exactly as how well-written object-oriented code
should be in terms of extensibility. 
We can add new classes such as \lstinline{Group} and \lstinline{HashJoin} without touching any of the existing code. 
This very desirable behavior demonstrates the effectiveness of the translation and can be observed throughout the case studies.
% (see the supplementary material for more examples of translated code and how it can be extended). 
Moreover, the extended code could be readily translated back to the original style, resulting in code similar to the hand-written one. 

%Yes, illustrated by our case studies on \lstinline{SqlFP} and \lstinline{BoolNormalizer}.

\paragraph{\systemname is applicable to a lot of existing code with moderate adaptions.} 
As demonstrated in the case studies, most adaptions are very straightforward and only involve moving the datatype parameter to the first parameter lists, 
fixing the name inconsistency of variables between case class and case clauses, or providing type information that \systemname currently cannot infer. 
The only slightly complicated case is when advanced forms of pattern matching are used.  
They need to be rewritten into the normal forms as illustrated.

% For a few cases that advanced forms of pattern matching are used, we need to rewrite them into the normal forms, as illustrated .
% The effort of adaption is roughly measured by counting the SLOC difference between the original and adapted programs.

\paragraph{\systemname transforms programs fast}
The case study programs are compiled using Scala 2.12.13, JDK 1.8.0\_292 and executed on a MacBook Pro (M1) with 8 cores and 16 GB memory.
We use ScalaMeter 0.8.2 \cite{scalameter} microbenchmark framework to measure the average time for 100 runs of transformation. 
As shown in the last column of \autoref{tbl:study}, the translation time is negligible and is positively correlated with the SLOC of the program being transformed.
\section{Discussion}{\label{sec:discussion}}

In this section, we discuss features not covered by \name and 
how \name can be applied to other multi-paradigm languages.

\subsection{Features beyond \name}\label{sec:morefeatures}
As a core calculus, there are a lot of language features not covered by \name.
As discussed in \autoref{sec:impl}, \name can be extended to handle more language features. 
Here, we discuss additional features that are not yet supported or cannot be supported in principle by \name.
% These features are either non-essential or do not yet have a clear correspondence in the other decomposition style or cannot be supported at all.

\paragraph*{Type tests/type casts}
Violating pure OOP principles, classes are sometimes used as types for type tests and type casts.
Type tests/type casts are frequently seen in \emph{binary methods}~\cite{BruEtAl96} for inspecting how the other object is constructed.
% A classical example of binary methods is the \lstinline{equals} operation for comparing 
% the structural equality of two objects. 
For example, to compare structural equality of two sets, one may change \lstinline{Insert} to:
\begin{lstlisting}
class Insert(val s: Set, val n: Int) extends Set { // val added for getters
  ...
  def equals(that: Set) =
    if (that.isInstanceOf[Insert]) {               // type test
      val thatInsert = that.asInstanceOf[Insert]   // type cast
      this.n == thatInsert.n && (this.s equals thatInsert.s)
    }
    else false
}
\end{lstlisting}
Besides being used as types, fields of \lstinline{Insert} are made public, exposing its internal representation.
Nonetheless, \lstinline{equals} can be rewritten without casts by adding a pair of \texttt{isInsert} and \lstinline{fromInsert} destructors ~\cite{emir2007matching}
for checking whether an object is created via \lstinline{Insert} and then extracting fields from an \lstinline{Insert} object:
% For example, \lstinline{Set} is augmented with a pair of destructors for class \lstinline{Insert}:
\begin{lstlisting}
trait Set {
  ...
  def isInsert: Boolean = false 
  def fromInsert: (Set,Int) = throw new RuntimeException
}  
\end{lstlisting}
In \lstinline{Insert}, \lstinline{isInsert} and \lstinline{fromInsert} are overridden accordingly: %so that \lstinline{equals} can be defined in terms of them:
\begin{lstlisting}
class Insert(s: Set, n: Int) extends Set {
  ...
  override def isInsert = true
  override def fromInsert = (s,n)
  def equals(that: Set) = 
    that.isInsert && s.equals(that.fromInsert._1) && (n == that.fromInsert._2)
}
\end{lstlisting}
This implementation still has issues like polluting the interface with implementation details
due to the fact that object-oriented decomposition discourages inspections on internal representation of other objects.
If such as an ability is a must, functional decomposition might be a better option.

\paragraph*{Complex patterns}
Functional decomposition allows inspection internal representation of data through pattern matching. 
Moreover, patterns can be nested or guarded by a a condition.  %multi-methods
Such complex patterns correspond to multi-methods \cite{chambers1992object}, which are neither available in pure OOP nor in many multi-paradigm languages like Scala.
Nevertheless, they can be rewritten into standard forms. 
As illustrated in \autoref{sec:pprint}, guarded patterns can be rewritten as if expressions 
and nested patterns can be replaced as multiple top-level patterns.
% To keep the code neutral on decomposition style, one can rewrite \lstinline{equals} using only top-level pattern matching.
If that becomes a burden, then it is time to fix the decomposition style to be functional. 

\paragraph*{Inheritance}
The original pure OOP \citep{cook2009} excludes inheritance since it is not a feature specific to OOP.
In \autoref{sec:overview}, we have shown a lightweight use of inheritance based on interfaces corresponding to wildcard pattern in functional decomposition.
There are other forms of inheritance such as class-based inheritance or even multiple inheritance.
The use of interface-based inheritance is considered less harmful than other forms as it neither introduces additional dependencies
nor causes the diamond problem. For other forms of inheritance that do not have a counterpart in functional decomposition, 
we rewrite them as explicit delegations. For example \lstinline{Neg} shown in \autoref{sec:intro} can be rewritten as:
\begin{lstlisting}
class Neg(e: Exp) extends Exp { def eval = new Sub(new Lit(0), e).eval }
\end{lstlisting}
% \systemname has some limited support 
% If complex class hierarchies with 

\paragraph{Mutable state} 
Mutable state is typically viewed as a object-oriented programming feature, which is indeed a feature from imperative programming. 
We can both write functional and object-oriented programs that manipulate mutable state. 
Supporting mutable state in \name is possible by allowing mutable fields in both generators and constructors.
We leave a formalization of \name with mutable state as future work. 
% One thing to note is that pattern variables create alias. We need to mutate the state via self.

\paragraph{Subtyping}
Subtyping is another feature that is often (mistakenly) considered as a object-oriented feature.
Nevertheless, object-oriented languages do have a better support for subtyping by allowing user-defined subtyping relations.
In OOP, interfaces can be extended with new destructors:
\begin{lstlisting}
trait ExtSet extends Set { def intersect(other: ExtSet): ExtSet }
\end{lstlisting}
The extended interface \lstinline{ExtSet} is a subtype of \lstinline{Set} and objects of \lstinline{ExtSet} can be passed as arguments to functions that accept \lstinline{Set}s.
However, a datatype is typically not allowed to be extended by another datatype. 
Even if it is allowed, e.g. in Scala:
\begin{lstlisting}
sealed trait ExtSet extends Set
case class Intersect(s1: ExtSet, s2: ExtSet) extends ExtSet
\end{lstlisting}
the extended datatype \lstinline{ExtSet} is treated as a subtype of \lstinline{Set} 
and existing consumers on \lstinline{Set} would be warned with a missing case for \lstinline{Intersect}.
However, \lstinline{ExtSet} should be indeed a \emph{supertype} of \lstinline{Set}.
Therefore, \name currently does not handle extended interfaces/datatypes as the goal of \name is to be applicable to existing multi-paradigm languages.

\paragraph{Parametric polymorphism}
Parametric polymorphism is a useful feature that is not yet supported by \name. 
Different forms of parametric polymorphism pose different level of challenges.
Supporting parametric data types and generic functions is relatively simple.
For instance, \lstinline{Set} can be defined as a parametric algebraic datatype with 
the element type captured as a type parameter:
\begin{lstlisting}[linerange=3-4]
sealed trait Set[A] 
case class Insert[A](s: Set[A], n: A) extends Set[A]
\end{lstlisting}%APPLY:GSET_FP
Accordingly, consumers on \lstinline{Set[A]} are also parameterized by \lstinline{A}.
In principle, this generic version of \lstinline{Set} can be transformed to parameterized class hierarchies and vice versa.
% \lstinputlisting[linerange=16-20]{../app/src/main/scala/GSetFP.scala}%APPLY:GCONTAINS_FP
However, supporting generalized algebraic datatypes (GADTs)~\citep{xi} would be problematic
as existing multi-paradigm languages may not support GADTs and their counterpart GAcoDTs~\cite{klaus2018}.
For example, GADTs are known to have soundness issues in Scala~\cite{giarrusso2013open}.
What makes it even tricker is that type parameters may interact with subtyping. For example,
type parameters be constrained with variance or bounds in Scala. 
Therefore, how to fully support parametric polymorphism remains an open problem.

\subsection{Applying \name to Scala 3}
The latest version of Scala, Scala 3, opens up many possibilities to further simplify \name.
Here, we summarize the key changes that are particularly relevant to our work.
% unified syntax for Scala

\paragraph*{Enumerations} In Scala 3, closed algebraic datatypes can be defined using \lstinline{enum}, e.g. 
% For example, the Scala 3 way of defining \lstinline{Set} in \autoref{setfp} is:
\begin{lstlisting}
enum Set:
  case Empty
  case Insert(s: Set, n: Int)
  case Union(s1: Set, s2: Set)
\end{lstlisting}
where constructors of \lstinline{Set} are defined in one place and explicit \lstinline{extends} clauses can be omitted.

\paragraph*{Extension methods}
Scala 3 introduces \emph{extension methods}, which allows adding methods to existing
types. Moreover, extension methods are invoked using the dot notation just like instance methods. 
Altogether, we can define consumers on datatype as extension methods, e.g. \lstinline{contains}:
\begin{lstlisting}
extension (self: Set) def contains(i: Int): Boolean = self match
  case Empty        => false
  case Insert(s,n)  => n == i || s.contains(i)
  case Union(s1,s2) => s1.contains(i) || s2.contains(i) 
\end{lstlisting}

This change brings two main advantages. First, extension methods
distinguish consumers from ordinary functions through the surface syntax.
As shown by the definition above, it clearly indicates that \lstinline{eval} is a consumer on \lstinline{Set}.
Second, recursive calls on \lstinline{contains} are written as \lstinline{s1.contains(i)} just like it is a destructor.

\paragraph*{Optional \lstinline{new}}
Scala 3 allows classes to be instantiated without \lstinline{new}, making it consistent with ordinary classes. 
Together with consumers defined as extension methods, client code for functional and object-oriented decomposition 
is made identical.  As a consequence, the bidirectional transformation can be simplified with only rules for definitions and the switch of decomposition style will not 
affect client code.

% \subsection{Duality in programming languages}

\begin{figure*}[t]
\begin{minipage}{.47\textwidth}
\begin{lstlisting}
interface Exp { fun eval(): Int }
class Lit(val n: Int): Exp {
  override fun eval(): Int { return n }
}
class Sub(val l: Exp, val r: Exp) : Exp {
  override fun eval(): Int {
    return l.eval() - r.eval()
  }
} 

fun main() {
  print(Sub(Lit(1),Lit(2)).eval())
}
\end{lstlisting}
\end{minipage}
\begin{minipage}{.52\textwidth}
\begin{lstlisting}
sealed class Exp 
data class Lit(val n: Int): Exp()
data class Sub(val l: Exp, val r: Exp): Exp()

fun main() {
  fun eval(self: Exp): Int {
    return when (self) {
      is Lit -> self.n
      is Sub -> eval(self.l) - eval(self.r)
    }
  }
  print(eval(Sub(Lit(1),(Lit(2)))))
}
\end{lstlisting}
\end{minipage}
\caption{Object-oriented decomposition vs functional decomposition in Kotlin}
\label{fig:kotlin}
\end{figure*}
\subsection{Applying \name to other multi-paradigm languages}
Although Scala is used throughout this paper for demonstration,
our approach is indeed not Scala-specific. The core features we are 
focusing on are commonly available in other multi-paradigm languages.
Modulo syntax differences, languages such as OCaml, F\#, Swift, Rust, and Kotlin are possible
candidates for \name.  As an example, \autoref{fig:kotlin} ports the simple arithmetic expression language shown in \autoref{decomposition} to Kotlin.
We can see that the Kotlin version looks very much like the Scala version.
Therefore, most of the rules shown in \autoref{sec:formalization} can be directly applied to Kotlin.
Of course, adjustments may be necessary for dealing with syntax differences. 
For example, Swift separates the field declarations and the initializer for generators, then the abstract syntax and 
rules named with \textsc{Gen} need to be changed accordingly.
There is an ongoing work that implements \name in Rust.
Preliminary results show that \name is also applicable to Rust, although Rust's type system, in particular smart pointers, poses new challenges that have not yet solved.

%TODO: mention rust implementation?

\begin{comment}
\begin{figure*}[t]
\begin{minipage}{.48\textwidth}
\begin{lstlisting}[language=caml]
type term = Lit of int 
          | Add of term * term
            
let rec eval term = match term with
  | Lit n -> n
  | Add (e1,e2) -> eval e1 + eval e2 

let e = eval (Add (Lit 1, Lit 2))
\end{lstlisting}
\end{minipage}
%
\begin{minipage}{.515\textwidth}
\begin{lstlisting}[language=caml]
class virtual term = object (self) 
  method virtual eval: int end
class lit(n:int) = object (self) inherit term 
  method eval = n end
class add (t1: term) (t2: term) = object (self) inherit term
  method eval = t1#eval + t2#eval end
let e = (new add (new lit 1) (new lit 2))#eval
\end{lstlisting}
\end{minipage}
\caption{OCaml}
\label{ocaml}
\end{figure*}
\end{comment}
\section{Related work}\label{sec:related}
% The relationship between objects and abstract data types.
% algebra / coalgebra(characteristic function)
\paragraph*{Expression Problem}
The Expression Problem dates back to \citet{Reynolds78}, who first pointed out that
user-defined types (abstract datatype) and procedural data structures are two complementary approaches to data abstraction in terms of extensibility.
\citet{cook1990object} further distilled the tradeoffs between the two data abstraction approaches and argued that objects are essentially procedural data structures.
% Matrix.% A program can be viewed as a matrix where data variants are rows and operations are columns.
% object-oriented decomposition is row-based and functional decomposition is column-based.
\citet{expPb} popularized the problem by coining the term ``Expression Problem'' and describing the requirements that a proper solution should meet.
Thereafter, many solutions to the Expression Problem have been proposed~\citep{zenger2001eadds,zenger-odersky2005,eptrivially16,oliveira2012extensibility,carette2009finally,swierstra2008data,Hofer:2008:PED:1449913.1449935,toplas}, to list a few.
% Most of the solutions rely on parametric polymorphism and/or subtyping, which \name currently does not support yet.
However, solutions to the Expression Problem typically introduce extra complexity, parameterization, and indirections 
that may cause performance penalty or require boilerplate code \citep{weixin2017,castor}. 
In contrast, \name allows programmers to write the code in their familiar style without performance penalty incurred by indirections.
% It would be interesting to investigate these Expression Problem solutions and their duality in future work.
% Functional decomposition approach shown in \autoref{} can be viewed sa an abstract data type
%
% \paragraph*{Defunctionalization and refunctionalization}
% Defunctionalization~\citep{reynolds1972definitional} is a trantformation that replaces higher-order functions by first-order programs with pattern matching.

% refunctionalization: there must be only one function that pattern matches on algebraic datatypes

\paragraph*{The duality of data and codata}
% The duality of language features, in  
There is a line of work studying the duality of data and codata~\citep{rendel2015,klaus2018,downen2019codata,popl2019,laforgue2017copattern}.
In particular, \citet{popl2019}'s work is most closely related and our formalization is greatly inspired by their work.
The major differences are that our transformation is type-directed whereas theirs is syntax-directed and 
our transformation can be applied to existing languages whereas theirs requires a new language design.
In their language design, destructors and consumers are in the same namespace with distinct names, which is quite different from the setting of existing languages
where classes have their own namespace and names can be overloaded.
As discussed in \autoref{sec:why-type}, type information is critical for transforming overloaded names correctly.
Their work additionally supports local (co)pattern matching.
However, specialized syntax is invented for explicitly capturing the environment and providing a name for transformation. 
This can hardly be enforced for existing languages and some convenience of local (co)pattern matching is lost.
Such local pattern matching and copattern matching are rewritten as top-level definitions in \name.
\citet{klaus2018} investigate the duality of generalized algebraic datatype and codatatype to their generalized versions.
How to port their formalization to existing languages remains a problem as there is inadequate support for these features in existing languages.
% consumer and generator like what we have done for \autoref{sec:case2}.
There are other transformation schemes between data to codata with a focus on compositionality~\citep{downen2019codata,laforgue2017copattern}.
Unlike \name, \citet{downen2019codata} compile data to codata using the \textsc{Visitor} pattern, which 
does not switch the dimension of extensibility.

\paragraph*{Multiple views of a program}
As a program can be written in multiple ways, \emph{intentional programming}~\citep{simonyi1995the} aims at capturing the intents of a programmer underneath the surface syntax.
The intents are a high-level representation of the program maintained in some database and the view, code, is generated on the fly. 
Following the idea of intentional programming, Decal~\citep{janzen2004programming} is a tool on an OOP language with open classes for allowing crosscutting concerns to be added in one place.
It maintains an internal representation of the program using a SQL engine and generates views from the internal representation. 
Decal users can either choose modules view (similar to functional decomposition in \name) or classes view (similar to object-oriented decomposition in \name) 
for editing and the changes to the view will be reflected to the internal representation.
Compared to Decal, \name does not need to maintain such an internal representation, simplifying the implementation. 
Moreover, unlike \name, Decal does not have a formalized transformation.
\section{Conclusion and future work}
In this paper, we have shown that restricted forms of functional and object-oriented decomposition are symmetric.
We propose a type-directed bidirectional transformation between functional and object-oriented decomposition in the \name calculus and proved the soundness of \name. 
Moreover, we have implemented \name in Scala called \systemname and conducted several case studies to demonstrate the applicability and effectiveness of \systemname.

In future work, we would like to explore more features such as mutable state and parametric polymorphism in \name.
How these additional features interact with existing features and affect the duality needs further investigation. 
We would also like to improve the usability of \systemname by reducing manual adaptions through automatic rewriting nested patterns, better type inference, etc. 
and develop similar tools on other multi-paradigm languages such as Rust. 
Another direction of future work is to mechanize our manual proofs in a theorem prover like Coq, where \citet{popl2019}'s work is a good start point.

% integrate \name with an IDE and provide decomposition style switching as a refactoring functionality.
% As discussed in \autoref{}, we would like to apply our transformation to other languages.

\bibliography{paper}
\appendix
\section{Additional formalization}\label{sec:additional-formalization}

In this section, we provide some additional formalization that, for brevity reasons, was left out 
of \autoref{sec:formalization}. Signatures for generators, constructors, consumers, and destructors
are collected using the following rules: 
% \begin{figure}
\begin{mathpar}
\framebox{$\sig(N) = T$}\\
\inferrule*
{\defs(C) = \class C D {Fun}}
{ \sig(C) = \overline{T} \rightarrow D }

\inferrule*
{\defs(C) = \ctr C D}
{ \sig(C) = \overline{T} \rightarrow D }

\inferrule*
{\defs(f) = \csmcase e}
{ \sig(f) = D \rightarrow \overline{T} \rightarrow T }
\\
\framebox{$\dtrType(f,D) = \overline{T} \rightarrow T$}\\
\inferrule*
{\defs(D) = \interface D {Dtr} \\ \dec \in \overline{Dtr}}
{ \dtrType(f,D) = \overline{T} \rightarrow T}
\end{mathpar}
% \caption{A few more preprocessing rules for computing $\Delta$}
% \label{fig:more-compute-delta}
% \end{figure}

\noindent \autoref{fig:skip-rest} gives the translation rules for skipping over code that were not captured in \autoref{fig:skip}.
Similarly to \autoref{fig:skip}, all the rules in \autoref{fig:skip-rest} work by only transforming the inner expressions.
\begin{figure}
\begin{mathpar}
\inferrule[Dt2Dt]
{\trans L T L'}
{ \trans {\sealed D; L} T {\sealed D; L'}}

\inferrule[Ctr2Ctr]
{ \trans L T L'}
{ \trans {\ctr{C}{D}; L} T \ctr{C}{D}; L'}

\inferrule[Fun2Fun]
{ \transCtx {\overline{x : T}} e T {e'}} % f : \overline{T} \rightarrow T is added to ctx 
{ \trans {\fun x e} {T} {\fun{e'}}}
%%{ \trans {\fun{e}} {\overline{T} \rightarrow T} {\fun{e'}}}

\inferrule[Csm2Csm]
{ \overline{\transCtx{\kwself: D, \overline{x:T}} {\case P e} {T} {\case P e'}} \\ \trans L T L'}
{ \trans {\csmcase e; L} {T} {\csmcase {e'}; L'}}
% $\inbracket{\kwdef \ f(x : D)(\overline{x : T}): T  = \overline{\kwcase\ P \Rightarrow e}} \text{ when } f \in \bigcup\limits_{D_i \in \ct}\consumer{D_i} = \varnothing$ \\
\end{mathpar}
% $\inbracket{\kwdef \ f(x : D)(\overline{x : T}): T  = \overline{\kwcase\ P \Rightarrow e}} \text{ when } f \in \bigcup\limits_{D_i \in \ct}\consumer{D_i} = \varnothing$ \\

% \framebox{$\trans e T e'$}
\begin{mathpar}

\inferrule[App2App]
{ \trans {e_1} D {e_1'} \\ \sig(f) = D \rightarrow \overline{T} \rightarrow T \in \Gamma \\ \overline{\trans {e_2} T {e_2'}}}
{ \trans{f(e_1)(\overline{e_2})}{T}{f(e_1')(\overline{e_2'})}}

\inferrule[Obj2Obj]
{ \sig(C) = \overline{T} \rightarrow D \\ \overline{\trans e T e'}}
{ \trans{C(\overline{e})}{D}{C(\overline{e'})}}

\end{mathpar}
\caption{Rules for types not selected for transformation}
\label{fig:skip-rest}
\end{figure}

\section{Lemmas and proofs} \label{sec:proofs}
In this section, we provide the proofs for the soundness theorems stated in \autoref{sec:soundness}, together with some helper lemmas.
To begin with, Lemmas~\ref{lemma:ctr2gen}-\ref{lemma:context} describe the relationship between the global context before and after the translation.  Lemmas \ref{lemma:csmBody} and \ref{lemma:dtrBody} describe the relation between the consumer and destructor bodies before and after the translation.

\begin{lemma}[Constructor to generator]\label{lemma:ctr2gen} if $\sig(C) = \overline{T} \rightarrow D$ and $C \in \ctrs{D}$ 
  then, after the translation, $\sig(C)' = \overline{T} \rightarrow D$ and $C \in \generator{D}'$. 
\end{lemma}
\begin{proof}
  Trivial by \textsc{Ctr2Gen}.
\end{proof}

\begin{lemma}[Generator to constructor]\label{lemma:gen2ctr} if $\sig(C) = \overline{T} \rightarrow D$ and $C \in \generator{D}$ 
  then, after the translation, $\sig(C)' = \overline{T} \rightarrow D$ and $C \in \ctrs{D}'$. 
\end{lemma}
\begin{proof}
  Trivial by \textsc{Gen2Ctr}.
\end{proof}

\begin{lemma}[Destructor to consumer]\label{lemma:dtr2csm} 
  if $f \in \dtr{D}$ and $\dtrType(f,D) = \overline{T} \rightarrow T$ 
  then, after the translation, $f \in \consumer{D}'$ and $\sig(f)' = D \rightarrow \overline{T} \rightarrow T$.
\end{lemma}
\begin{proof}
  Trivial by \textsc{It2Dt} and \textsc{Dec2Csm}.
\end{proof}

\begin{lemma}[Consumer to destructor]\label{lemma:csm2dtr} 
  if $f \in \consumer{D}$ and $\sig(f) = D \rightarrow \overline{T} \rightarrow T$\\
  then, after the translation, $f \in \dtr{D}'$ and $\dtrType(f,D)' = \overline{T} \rightarrow T$. 
\end{lemma}
\begin{proof}
  Trivial by \textsc{Dt2It} and \textsc{Csm2Dec}.
\end{proof}

\begin{lemma}[Global context translation]\label{lemma:context} if $\Delta$ is the global context collected through the preprocessing of $P$, and $\Delta'$ is the context collected through the preprocessing of $P'$ such that $\trans P T P'$, then: $\textsc{It} = \textsc{Dt}'$ and $\textsc{Dt} = \textsc{It}'$
and $\forall D\in \textsc{It}.\textsc{Dtr}(D)=\textsc{Csm}(D)' \wedge \textsc{Gen}(D)=\textsc{Ctr}(D)'$ and $\forall D\in \textsc{Dt}.\textsc{Csm}(D)=\textsc{Dtr}(D)' \wedge \textsc{Ctr}(D)=\textsc{Gen}(D)'$. We write $\Delta \leadsto \Delta'$.
\end{lemma}
\begin{proof}
Trivial by \textsc{It2Dt}, \textsc{Dt2It} and Lemmas~\ref{lemma:ctr2gen}-\ref{lemma:csm2dtr}.
\end{proof}

%% Next, we provide the proof for \autoref{theorem:tp}.

%% \begin{proof}
%%   By induction.

%% \begin{itemize}
%% \item Case \appCtr

%% \begin{longtable}[l]{ll}
%% $\sig(C)' = \overline{T} \rightarrow D$ and $C \in \generator{D}'$ & By \autoref{lemma:ctr2gen}\\
%% $\overline{\trans{e'}{T}{e}}$ &  By i.h.\\
%% $\trans{\new C {e'}}{D}{C(\overline{e})}$ & By \textsc{New2Obj}\\
%% \end{longtable}

%% \item Case \newGen

%% \begin{longtable}[l]{ll}
%% $\sig(C)' = \overline{T} \rightarrow D$ and $C \in \ctrs{D}'$  & By \autoref{lemma:gen2ctr} \\
%% $\overline{\trans{e'}{T}{e}}$ &  By i.h.\\
%% $\trans{C(\overline{e'})}{D}{\new C {e}}$ & By \textsc{Obj2New}\\
%% \end{longtable}
%% \newpage
%% \item Case \selDtr

%% \begin{longtable}[l]{ll}
%% $f \in \consumer{D}'$ and $\sig(f)' = D \rightarrow \overline{T} \rightarrow T$ & By \autoref{lemma:dtr2csm}\\
%% $\trans{e_1'}{T}{e_1}$ &  By i.h.\\
%% $\overline{\trans{e_2'}{T}{e_2}}$ &  By i.h.\\
%% $\trans{f(e_1')(\overline{e_2'})}{T}{e_1.f(\overline{e_2})}$ & By \textsc{App2Sel}\\
%% \end{longtable}

%% \item Case \appCsm

%% \begin{longtable}[l]{ll}
%% % \overline{\trans {e_2} T {e_2'}}
%% $\trans {e_1'} D {e_1}$ & By i.h. \\
%% $\overline{\trans {e_2'} T {e_2}}$ & By i.h.\\
%% $f \in \dtr{D}'$ and $\dtrType(f,D) = \overline{T} \rightarrow T$ & By \autoref{lemma:csm2dtr}\\
%% $\trans{f(e_1')(\overline{e_2'})}{T}{e_1.f(\overline{e_2})}$ & By \textsc{Sel2App}\\
%% \end{longtable}
%% \end{itemize}
%% \end{proof}

\cristina{Modified a bit the next two lemmas.}
\begin{lemma}[Destructor to consumer case] \label{lemma:csmBody} if $C \in \generator{D}$, $f \in \dtr{D}$, $\dtrBody(f,C) = (\overline{y},\overline{x},e)$ and  $\trans e T {e'}$
  then, after the translation, $C \in \ctrs{D}'$, $f \in \consumer{D}'$ and $\csmBody(f,C) = (\overline{y},\overline{x},[\kwthis \mapsto \kwself]e')$.
\end{lemma}
\begin{proof}
By induction. 
\begin{itemize}
\item Case \dtrGen.

\begin{longtable}[l]{ll}
$\dtrToCase{\classf {C} {y} {T} {D} {\overline{Fun}}} {\case {C(\overline{y})} [\kwthis \mapsto \kwself]e'}$ & By \textsc{Fun2Case}\\
% $\case{C(\overline{y})}{e'} \in \overline{\case P e}$ & By Dec2Csm\\
% $\case{C(\overline{y})}{[\kwthis \mapsto x]e'}$ & By Fun2Case\\
% $\dtrToCase {\classf C y D {\overline{Fun}}} {\case{C(\overline{y})}{[\kwthis \mapsto x]e'}}$ & By Fun2Case \\
% $C \in \generator{D}$ & Given\\
$\csmBody(f,C) = (\overline{y},\overline{x},[\kwthis \mapsto \kwself]e') $ & By \textsc{CsmCtr}
\end{longtable}

\item Case \dtrIt:\\
\begin{longtable}[l]{ll}
  $\dtrToCsm {\dec x = e}$\\
  \qquad\qquad $\csmcase e;\case{\_}{[\kwthis \mapsto \kwself]e'}$ & By \textsc{Fun2Csm}\\
$\csmBody(f,C) = (\varnothing, \overline{x},{[\kwthis \mapsto \kwself]e'})$ & By \textsc{CsmDt}\\
\end{longtable}
\end{itemize}
\end{proof}

\begin{lemma}[Consumer to destructor case] \label{lemma:dtrBody} if $C \in \ctrs{D}$, $f \in \consumer{D}$, $\csmBody(f,C) = (\overline{y},\overline{x},e)$ and $\trans e T e'$
  then, after the translation, $C \in \generator{D}'$, $f \in \dtr{D}'$,  $\dtrBody(f,C) = (\overline{y},\overline{x},[\kwself\mapsto \kwthis]e')$
\end{lemma}

\begin{proof}
% $\defs(f) = \csmcase e$\\
By induction. 
\begin{itemize}
\item Case \csmCtr
\begin{longtable}[l]{ll}
$\csmToDtr { \csmcase e} {\dec x = [ \kwself \mapsto \kwthis]e'}$ & By \textsc{Case2Fun}\\
$\dtrBody(f,C) = (\overline{y}, \overline{x},{[\kwself \mapsto \kwthis]e'})$ & By \textsc{DtrGen}\\
\end{longtable}
\item Case \csmDt
\begin{longtable}[l]{ll}
$\Delta; \Gamma \vdash \csmcase e  \leadsto \dec x = [ \kwself \mapsto \kwthis]e'$ & By \textsc{Csm2Fun}\\
$\dtrBody(f,C) = (\varnothing, \overline{x},{[\kwself \mapsto \kwthis]e'})$ & By \textsc{DtrIt}\\
\end{longtable}
\end{itemize}
\end{proof}

  \subsection{Proof for Theorem~\ref{theorem:preservation}}

  \begin{proof}
By rule induction on the given typing derivation.
%\todo{In the proof I mention standard substitution rules. I don't think we need to actually provide them..}

\begin{itemize}
%\item \todo{Add rule for $x$?} 
\item Case {\sc Sel2App}. Given $e_1.f(\overline{e_2})$, either {\sc E-Congr} or {\sc E-Dtr} apply.
  \begin{enumerate}
  \item If {\sc E-Congr} applies, then either $\eval {e_1} {e_1'}$, or, there exists some
  $e_2$ in $\overline{e_2}$ such that $\eval {e_2} {e_2'}$.
  \begin{itemize}
  \item For the first case, by the induction hypothesis, if $\Delta; \boldsymbol{\cdot} \vdash e_1 : T_1 \leadsto \_$, then
  $\Delta; \boldsymbol{\cdot} \vdash e_1' : T_1 \leadsto \_$. Consequently, it must be the case that if $\Delta; \boldsymbol{\cdot} \vdash e_1.f(\overline{e_2}) : T \leadsto \_$,
  then $\Delta; \boldsymbol{\cdot} \vdash e_1'.f(\overline{e_2}) : T \leadsto \_$.
  \item Similarly, for the second case, if $\Delta; \boldsymbol{\cdot} \vdash e_2 : T_2 \leadsto \_$, then
  $\Delta; \boldsymbol{\cdot} \vdash e_2' : T_2 \leadsto \_$ by the induction hypothesis. Consequently, if $\Delta; \boldsymbol{\cdot} \vdash e_1.f(\overline{e_2}) : T \leadsto \_$,
  then $\Delta; \boldsymbol{\cdot} \vdash e_1.f(\overline{e_2'}) : T \leadsto \_$. 
  \end{itemize}
  \item If {\sc E-Dtr} applies, then the conclusion results from {\sc Fun2Csm}, destructor lookup and standard substitution rules.\\
  \end{enumerate}
\item Case {\sc App2Sel}. Given $f(e_1)(\overline{e_2})$, either {\sc E-Congr} or {\sc E-Csm} apply.
  \begin{enumerate}
  \item If {\sc E-Congr} applies, then either $\eval {e_1} {e_1'}$, or, there exists some
  $e_2$ in $\overline{e_2}$ such that $\eval {e_2} {e_2'}$.
  \begin{itemize}
  \item For the first case, by the induction hypothesis, if $\Delta; \boldsymbol{\cdot} \vdash e_1 : T_1 \leadsto \_$, then
  $\Delta; \boldsymbol{\cdot} \vdash e_1' : T_1 \leadsto \_$. Consequently, it must be the case that if $\Delta; \boldsymbol{\cdot} \vdash f(e_1)(\overline{e_2}) : T \leadsto \_$,
  then $\Delta; \boldsymbol{\cdot} \vdash f(e_1')(\overline{e_2}) : T \leadsto \_$.
  \item Similarly, for the second case, if $\Delta; \boldsymbol{\cdot} \vdash e_2 : T_2 \leadsto \_$, then
  $\Delta; \boldsymbol{\cdot} \vdash e_2' : T_2 \leadsto \_$ by the induction hypothesis. Consequently, if $\Delta; \boldsymbol{\cdot} \vdash f(e_1)(\overline{e_2}) : T \leadsto \_$,
  then $\Delta; \boldsymbol{\cdot} \vdash f(e_1)(\overline{e_2'}) : T \leadsto \_$. 
  \end{itemize}
  \item If {\sc E-Csm} applies, then the conclusion results from {\sc Csm2Fun}, consumer lookup and standard substitution rules.\\
  \end{enumerate}
\item Case {\sc Obj2New}. Given $C(\overline{e})$, either {\sc E-Congr} or {\sc E-Ctr} apply.
\begin{enumerate}
\item If {\sc E-Congr} applies, then there exists some $e$ in $\overline{e}$ such that $\eval {e} {e'}$. By the induction hypothesis,
  if $\Delta; \boldsymbol{\cdot} \vdash e : T_1 \leadsto \_$, then
  $\Delta; \boldsymbol{\cdot} \vdash e' : T_1 \leadsto \_$. Consequently, it must be the case that if $\Delta; \boldsymbol{\cdot} \vdash C(\overline{e}) : T \leadsto \_$,
  then $\Delta; \boldsymbol{\cdot} \vdash C(\overline{e'}) : T \leadsto \_$.
\item If {\sc E-Ctr}, then the conclusion follows from the typing rule {\sc Obj}, where we ignore the translation part (as no translation is needed for $obj$):

\begin{mathpar}  
\inferrule[Obj]
{ \sig(C) = \overline{T} \rightarrow D }
{ \trans{obj(C,\overline{v})}{D}{\_}}
\end{mathpar}

\end{enumerate}  

\item Case {\sc New2Obj}. Given $\newSeq{C}{e}$, either {\sc E-Congr} or {\sc E-New} apply.
\begin{enumerate}
\item If {\sc E-Congr} applies, then there exists some $e$ in $\overline{e}$ such that $\eval {e} {e'}$. By the induction hypothesis,
  if $\Delta; \boldsymbol{\cdot} \vdash e : T_1 \leadsto \_$, then
  $\Delta; \boldsymbol{\cdot} \vdash e' : T_1 \leadsto \_$. Consequently, it must be the case that if $\Delta; \boldsymbol{\cdot} \vdash \newSeq{C}{e} : T \leadsto \_$,
  then $\Delta; \boldsymbol{\cdot} \vdash \newSeq{C}{e'} : T \leadsto \_$.
\item If {\sc E-New}, then the conclusion follows from the typing rule {\sc Obj} above.
\end{enumerate}  
\end{itemize}
\end{proof}

  \subsection{Proof for Theorem~\ref{theorem:progress}}

  \begin{proof}
By rule induction on the given typing derivation.
%\todo{In the proof I mention standard substitution rules. I don't think we need to actually provide them..}

\begin{itemize}
%\item \todo{Add rule for $x$?} 
\item Case {\sc Sel2App}. Given $e_1.f(\overline{e_2})$, either  {\sc E-Congr} applies, or $e_1 = obj(C,\overline{v_1})$ and $\overline{e_2} = \overline{v_2}$. By {\sc E-Dtr}, the latter scenario evaluates to $[\kwthis \mapsto obj(C,\overline{v_1}),\overline{y} \mapsto \overline{v_1}, \overline{x} \mapsto \overline{v_2} ]e_3$, where $\dtrBody(f,C) = (\overline{y},\overline{x},e_3)$. By the induction hypothesis and standard substitution rules $[\kwthis \mapsto obj(C,\overline{v_1}),\overline{y} \mapsto \overline{v_1}, \overline{x} \mapsto \overline{v_2} ]e_3$ either evaluates to $e_3'$, or it's a value. \\
\item Case {\sc App2Sel}. Given $f(e_1)(\overline{e_2})$, either {\sc E-Congr} applies, or $e_1 = obj(C,\overline{v_1})$ and $\overline{e_2} = \overline{v_2}$. By {\sc E-Csm}, the latter scenario evaluates to $[\kwself \mapsto obj(C,\overline{v_1}),\overline{y} \mapsto \overline{v_1}, \overline{x} \mapsto \overline{v_2} ]e_3$, where $\csmBody(f,C) = (\overline{y},\overline{x},e_3)$. By the induction hypothesis and standard substitution rules $[\kwself \mapsto obj(C,\overline{v_1}),\overline{y} \mapsto \overline{v_1}, \overline{x} \mapsto \overline{v_2} ]e_3$ either evaluates to $e_3'$, or it's a value.\\
\item Case {\sc Obj2New}. Given $C(\overline{e})$, either {\sc E-Congr} applies, or $\overline{e}$ is $\overline{v}$. By {\sc E-Ctr}, the latter evaluates to $obj(C,\overline{v})$, which is a value.\\
\item Case {\sc New2Obj}. Given $\newSeq{C}{e}$, either {\sc E-Congr} applies, or $\overline{e}$ is $\overline{v}$. By {\sc E-New}, the latter evaluates to $obj(C,\overline{v})$, which is a value.

\end{itemize}
\end{proof}

\subsection{Proof for \autoref{theorem:tp}}

%% \begin{theorem}[Syntax and type preservation]\label{theorem:tp} if $\trans P T P'$ and $\Delta \leadsto \Delta'$, then $\transGeneric {\Delta'; \Gamma} {P'} T P$. 
%% \end{theorem}

\begin{proof}
  By induction.

\begin{itemize}
\item Case \appCtr

\begin{longtable}[l]{ll}
$\sig(C)' = \overline{T} \rightarrow D$ and $C \in \generator{D}'$ & By Lemma \ref{lemma:ctr2gen}\\
$\overline{\trans{e'}{T}{e}}$ &  By i.h.\\
$\trans{\new C {e'}}{D}{C(\overline{e})}$ & By \textsc{New2Obj}\\
\end{longtable}

\item Case \newGen

\begin{longtable}[l]{ll}
$\sig(C)' = \overline{T} \rightarrow D$ and $C \in \ctrs{D}'$  & By Lemma \ref{lemma:gen2ctr} \\
$\overline{\trans{e'}{T}{e}}$ &  By i.h.\\
$\trans{C(\overline{e'})}{D}{\new C {e}}$ & By \textsc{Obj2New}\\
\end{longtable}

\item Case \selDtr

\begin{longtable}[l]{ll}
$f \in \consumer{D}'$ and $\sig(f)' = D \rightarrow \overline{T} \rightarrow T$ & By Lemma \ref{lemma:dtr2csm}\\
$\trans{e_1'}{T}{e_1}$ &  By i.h.\\
$\overline{\trans{e_2'}{T}{e_2}}$ &  By i.h.\\
$\trans{f(e_1')(\overline{e_2'})}{T}{e_1.f(\overline{e_2})}$ & By \textsc{App2Sel}\\
\end{longtable}

\item Case \appCsm

\begin{longtable}[l]{ll}
% \overline{\trans {e_2} T {e_2'}}
$\trans {e_1'} D {e_1}$ & By i.h. \\
$\overline{\trans {e_2'} T {e_2}}$ & By i.h.\\
$f \in \dtr{D}'$ and $\dtrType(f,D) = \overline{T} \rightarrow T$ & By Lemma \ref{lemma:csm2dtr}\\
$\trans{f(e_1')(\overline{e_2'})}{T}{e_1.f(\overline{e_2})}$ & By \textsc{Sel2App}\\
\end{longtable}
\end{itemize}
\end{proof}

  \subsection{Proof for \autoref{sp}}
\begin{proof}
We'll start with the proof for condition (1). We proceeed by induction on the translation rules.

\begin{itemize}
\item Case {\sc App2Sel}. Given that $P$ doesn't diverge, we must have $e_1 \longrightarrow^* v_1$ and $ \overline{e_2} \longrightarrow^* \overline{v_2}$ (for the latter, we abuse the notation
to mean that each individual $e_2$ evaluates to a corresponding $v_2$). Then, by induction hypothesis
$e_1' \longrightarrow^* v_1$ and $ \overline{e_2'} \longrightarrow^* \overline{v_2}$.
The conclusion follows from Lemma \ref{lemma:dtrBody} and rules {\sc E-Csm} and {\sc E-Dtr} in \autoref{semantics}.
\item Case {\sc Sel2App}. Given that $P$ doesn't diverge, we must have $e_1 \longrightarrow^* v_1$ and $ \overline{e_2} \longrightarrow^* \overline{v_2}$. Then, by induction hypothesis
$e_1' \longrightarrow^* v_1$ and $ \overline{e_2'} \longrightarrow^* \overline{v_2}$.
The conclusion follows from Lemma \ref{lemma:csmBody} and rules {\sc E-Dtr} and {\sc E-Csm} in \autoref{semantics}.
\item Case {\sc Obj2New}. Given that $P$ doesn't diverge, we must have $ \overline{e} \longrightarrow^* \overline{v_1}$. Then, by induction hypothesis
$ \overline{e'} \longrightarrow^* \overline{v_1}$. From rules {\sc E-Ctr} and {\sc E-New} in \autoref{semantics}, it follows that both $C(\overline{v_1})$
and $\newSeq{C}{v_1}$ evaluate to the same value.
\item Case {\sc New2Obj}. Follows exactly the same pattern as the case for {\sc Obj2New}.
\end{itemize}

Regarding the proof for condition (2), if program $P$ can diverges, %because of infinite recursion.
then let $T$ be an infinite trace of $P$. Given the rules in \autoref{semantics} and the fact $P$ contains a finite number of instructions,
one of the following situations must happen:
\begin{itemize}
\item[(i)] There exists a method invocation  $e_1.f(\overline{e_2})$ that appears an infinite number of times in $T$.
According to {\sc Sel2App}, each method invocation $e_1.f(\overline{e_2})$ is translated to a corresponding
function application $f(e_1')(\overline{e_2'})$ in $P'$.
By the induction hypothesis, we have that if
$e_1 \longrightarrow^* v_1$ and $ \overline{e_2} \longrightarrow^* \overline{v_2}$,
then $e_1' \longrightarrow^* v_1$ and $ \overline{e_2'} \longrightarrow^* \overline{v_2}$.
Then, the execution of $P'$ must contain an infinite number of function applications $f(e_1')(\overline{e_2'})$.
\item[(ii)] There exists a function application  $f(e_1)(\overline{e_2})$ that appears an infinite number of times in $T$.
According to {\sc App2Sel}, each function application $f(e_1)(\overline{e_2})$
is translated to a corresponding method invocation $e_1'.f(\overline{e_2'})$ 
 in $P'$.
By the induction hypothesis, we have that if
$e_1 \longrightarrow^* v_1$ and $ \overline{e_2} \longrightarrow^* \overline{v_2}$,
then $e_1' \longrightarrow^* v_1$ and $ \overline{e_2'} \longrightarrow^* \overline{v_2}$.
Then, the execution of $P'$ must contain an infinite number of method invocations $e_1'.f(\overline{e_2'})$.
\end{itemize}
\end{proof}

% \begin{acks}
% \end{acks}

\end{document}